\documentclass[11pt]{article}
\usepackage{fullpage}
\usepackage[numbers]{natbib}
\usepackage[utf8]{inputenc}
\usepackage{amsmath,amsthm,amssymb}
\usepackage{thmtools} 
\usepackage{thm-restate}
\usepackage{amsfonts,bm,bbold,dsfont}
\usepackage{enumitem}
\usepackage[normalem]{ulem}
\usepackage{xspace}
\usepackage[usenames,dvipsnames]{xcolor}
\usepackage[colorlinks=true,citecolor=ForestGreen,linkcolor=ForestGreen,urlcolor=NavyBlue]{hyperref}
\usepackage{algorithm}
\usepackage{colortbl}
\usepackage{caption}
\usepackage{epigraph}

\usepackage[nameinlink,capitalise]{cleveref}

\usepackage{comment} 
\usepackage{subcaption}
\usepackage{multirow}
\usepackage{pgf,tikz,pgfplots}
\usetikzlibrary{calc,patterns,intersections,pgfplots.fillbetween}
\pgfplotsset{compat=newest}

\usepackage[suppress]{color-edits}
\addauthor{DM}{red}

\addauthor{MF}{blue}

\addauthor{SM}{purple!80!black}

\addauthor{AE}{magenta}

\addauthor{KG}{orange}

\renewcommand{\v}{\mathbf{v}}
\newcommand{\s}{\mathbf{s}}
\newcommand{\E}{\mathop{\mathds{E}}}

\newcommand{\sig}[1][]{#1{s}}
\newcommand{\val}[1][]{#1{v}}
\newcommand{\alloc}[1][]{#1{x}}
\newcommand{\price}[1][]{#1{p}}
\newcommand{\sigs}[1][]{\mathbf{\sig[#1]}}
\newcommand{\sigspace}{\mathbf{S}}
\newcommand{\vals}[1][]{\mathbf{\val[#1]}}
\newcommand{\allocs}[1][]{\mathbf{\alloc[#1]}}
\newcommand{\prices}[1][]{\mathbf{\price[#1]}}
\newcommand{\low}[3][]{\underline{\val[#1]}_{#2}(\sigs[#1]_{-#3})}

\newcommand{\selfbounding}{\textit{self-bounding}\xspace}

\newcommand{\critical}{\textit{critical}\xspace}
\newcommand{\sos}{\textit{SOS}\xspace}
\newcommand{\logdag}{\log_2^\dagger}
\newcommand{\remove}[1]{}

\DeclareMathOperator*{\argmax}{\mathrm{argmax}}

\newtheorem{definition}{Definition}[section]
\newtheorem{lemma}{Lemma}[section]
\newtheorem{theorem}{Theorem}[section]

\newtheorem{example}{Example}[section]
\newtheorem{proposition}{Proposition}[section]

\newtheorem{observation}{Observation}[section]

\allowdisplaybreaks

\title{Constant Approximation for Private Interdependent Valuations\thanks{The work of A. Eden was supported by the Golda Meir Fellowship, The work of M. Feldman, S. Mauras and D. Mohan was partially supported by the European Research Council (ERC) under the European Union's Horizon 2020 research and innovation program (grant agreement No. 866132), by an Amazon Research Award, and by the NSF-BSF (grant number 2020788). The work of K. Goldner was supported by NSF Award CNS-2228610 and a Shibulal Family Career Development Professorship.}}

	\author{
		Alon Eden
		\thanks{The Hebrew University; {\tt alon.eden@mail.huji.ac.il}}
		\and
		Michal Feldman
		\thanks{Tel Aviv University; {\tt michal.feldman@cs.tau.ac.il}}
            \and
            Kira Goldner
            \thanks{Boston University; {\tt goldner@bu.edu}}
            \and
            Simon Mauras
            \thanks{Tel Aviv University; {\tt smauras@tauex.tau.ac.il}}
            \and
            Divyarthi Mohan
            \thanks{Tel Aviv University; {\tt divyarthim@tau.ac.il}}
	}

\begin{document}

\thispagestyle{empty}
\maketitle
\begin{abstract}
The celebrated model of auctions with interdependent valuations, introduced by Milgrom and Weber in 1982, has been studied almost exclusively under private signals $s_1, \ldots, s_n$ of the $n$ bidders and public valuation functions $v_i(s_1, \ldots, s_n)$. 
Recent work in TCS has shown that this setting admits a constant approximation to the optimal social welfare if the valuations satisfy a natural property called submodularity over signals (SOS).
More recently, Eden et al. (2022) have extended the analysis of interdependent valuations to include settings with private signals and \emph{private valuations}, and established $O(\log^2 n)$-approximation for SOS valuations. In this paper we show that this setting admits a {\em constant} factor approximation, settling the open question raised by Eden et al. (2022).
\end{abstract}

\newpage
\setcounter{page}{1}

\setlength\epigraphwidth{7.5cm}
\setlength\epigraphrule{0pt}
\epigraph{
    \emph{``You can fool some of the people all of the time, and all of the people some of the time, but you can not fool all of the people all of the time.''}}{{---Abraham Lincoln}}

\section{Introduction}

{The interdependent values model captures auction scenarios where each bidder has some partial information about the good for sale}, but their value for the good depends also on the information held by other bidders~\cite{MilgromWeber82,wilson1969communications}. This captures many realistic scenarios such as the selling of a natural resource {(e.g., oil)} of unknown value, art auctions, and ad-auctions of online impressions, among many others. This model has been widely studied in the economic literature, with its importance being recognized by the 2020 Nobel Prize in Economics~\cite{nobel2021considerations}. In  this model, each bidder $i$ possesses a \textit{private} signal: a real number $s_i$ (e.g., the estimate that the bidder has for the amount of oil in the auctioned oil field). The bidder also possesses  a \textit{public} valuation function $v_i(s_1,\ldots,s_n)$ which maps the {signals of {\em all} bidders---one's own signal, as well as others'---into a value for the item for sale} (e.g., a bidder's value for the expected amount of oil in the field given all bidders' information).

Previous work in economics has found that this intricate setting gives rise to many impossibility results, and good design is possible only in very restricted cases~\cite{maskin1992,DM00,JM69, ausubel1999generalized}.
More recently, the EconCS community has put effort toward circumventing these impossibilities via the {algorithmic} lens of approximation (e.g.,~\cite{RTC, CFK, EFFG, EdenFFGK19, AmerTC21,EdenGZ22,LuSZ22, CohenFMT23,gkatzelis2021prior,ChenEW}). 

{A major breakthrough toward positive welfare guarantees comes from ~\citet{EdenFFGK19}, who devise a 4-approximation mechanism 
for valuation functions that
satisfy a property they refer to as \textit{Submodular-over-Signals} (or \sos). 
\sos is a natural generalization of the submodularity property of set functions.
}
Roughly speaking, a valuation function $v(\cdot)$ is \sos if, for a signal profile $\s_{-j}$ for all bidders but $j$, when the signals $\s_{-j}$ are \emph{lower}, then an \emph{increase} in signal $s_j$ has a \emph{larger} effect on the valuations. In other words, information (signals) exhibit decreasing marginal returns. The \sos condition captures {many natural  settings including} most {scenarios} studied in the literature, such as mineral-rights auctions and art auctions.

Only more recently, \citet{EdenGZ22} have studied the case where valuation functions are assumed to be private (unknown to the seller or other bidders), just like signals are. This reflects the fact that in many real-world settings, 
{individuals' private information encompasses both their partial information about the good for sale, as well as the way they aggregate everyone's private information into a value. 
For instance, in the oil field example, an oil company's signal may be an estimate of the amount of oil, while their valuation function may reflect the company's estimated production cost, which may impact the profitability of the oil field.
There is no reason to assume that the signal---the estimated amount of oil in this case---is less private than the valuation function---the estimated production cost in this case. 
Indeed, while public valuations can be shown to be easier to handle, in the oil example, as well as in many additional auction settings, it is much more realistic to assume that both signals and valuation functions are private information.

As \citet{EdenGZ22} show, private valuation functions pose a much greater challenge than public ones. 
In particular, the single-crossing condition\footnote{Roughly speaking, a valuation function is single-crossing if $i$'s signal $s_i$ affects $i$'s valuation at least as much as it does any other bidder.}, which enables full efficiency in single-item auctions with public valuations and private signals, is rendered useless in settings where the valuations are private as well, and cannot guarantee more than the trivial $n$-approximation. 
On the positive side, they devise an $O(\log ^2 n)$-approximation mechanism when the valuations are \sos. \citet{EdenGZ22} left the following question unresolved:
\begin{quotation}
\emph{``Is there a mechanism that achieves a constant-factor approximation to the optimal welfare under private signals and private valuations?"}
\end{quotation}

We answer this question in the affirmative. Informally, our main result is the following.
    
\vspace{2mm}
\noindent\textbf{Main Theorem.} There exists a {polynomial time} truthful mechanism that gives  a constant-factor approximation to the optimal welfare in a single-item auction with private interdependent valuations that satisfy the \sos condition.  
\vspace{2mm}

Our main result extends in several ways: First, it applies to non-monotone \sos valuations. 
Second, it extends beyond single-item auctions, to settings with unit-demand valuations over multiple identical items.

Notably, our mechanism is randomized. 
This is inevitable, as even in the case of public \sos valuations, one cannot guarantee any approximation with deterministic mechanisms \cite{EdenFFGK19}. 
Moreover, turning to approximation is inevitable, as even in the case of public \sos valuations, one cannot get better than $2$-approximation even with a randomized mechanism \cite{EdenFFGK19}.

Our results are derived by introducing a new hierarchy of valuations which we term $d$-\selfbounding valuations, where each valuation profile is parameterized by $d\in \{1,\ldots, n\}$. The mechanism we devise gives a tight $\Theta(d)$-approximation for $d$-\selfbounding valuations. Our main results then follow by showing that monotone \sos valuations are $1$-\selfbounding and non-monotone \sos valuations are $2$-\selfbounding. 

\subsection{Related Work}
The two immediate precursors of this work are~\citet{EdenFFGK19}, which introduces \sos valuations for interdependent settings and brings them into the context of combinatorial auctions, and~\citet{EdenGZ22}, which is the first work to study interdependent valuations with private valuation functions. On the combinatorial front, \citep{EdenFFGK19} devises a random-sampling version of the VCG mechanism that obtains a $4$-approximation for {\emph{public}} \sos valuations that satisfy an additional separability condition.  For \emph{private} valuation functions under the \sos condition, \citet{EdenGZ22} show a $O(\log^2 n)$-approximation mechanism in the single-item setting. {The same paper~\cite{EdenGZ22} also considers a restricted setting, where the valuation functions depend on the signals of at most a constant number of bidders, and provides a constant-approximation mechanism for this case.} Related is \citet{DM00}, who study the interdependent setting when the seller is unaware of bidders' valuation functions, but crucially, the bidders do know each others' valuation functions. They devise a mechanism where the bidders bid a complicated contingent bidding function which maps the bids of other bidders to a bid for the bidder. They show that under single-crossing-type conditions, there is a fully-efficient equilibrium. 

A recent burst of work from the EconCS community has investigated interdependent value auctions in a variety of settings. \citet{RTC} study prior-independent mechanisms for revenue maximization. \citet{CFK} design approximately-optimal mechanisms for revenue maximization while trying to minimize the assumptions needed. \citet{EFFG} relax the single-crossing condition and study approximately-optimal welfare-maximizing auctions. \citet{AmerTC21} and \citet{LuSZ22} improve the approximation guarantees for single-item auctions under the \sos condition. \citet{EdenFTZ21} study the PoA of simple, non-truthful mechanisms. \citet{CohenFMT23} study combinatorial public projects in interdependent settings.  

Our work  introduces classes of valuation functions over signals analogous to the combinatorial valuation functions studied by~\cite{LehmannLN06}. The combinatorial valuations were proven useful in devising nearly-optimal mechanisms for welfare~\cite{Dobzinski21,Assadi019,DobzinskiNS06} and revenue~\cite{CaiZ17, RubinsteinW18}, as well as nearly simple, non-truthful, nearly-optimal mechanisms~\cite{CKS16, SyrgkanisT13, FeldmanFGL13}.  

\subsection{Organization}

In \Cref{sec:prelim}, we present our model and main definitions; specifically, in \Cref{sec:prelim-model} we present the interdependent values model, and give sufficient conditions for a truthful mechanism and in \Cref{sec:prelim-defs} we present the main properties of valuations used in this paper. 
In \Cref{sec:ideas-techniques}, we discuss the main ideas and intuition of our mechanism by presenting a na\"ive attempt and discussing the obstacles to this approach.
In \Cref{sec:d-sb-mech}, we present and prove our main result: a truthful $\Theta(d)$-approximation mechanism to $d$-\selfbounding valuations. Finally, in \Cref{sec:multi-unit}, we extend our result to the case of multiple identical items and unit-demand bidders.

\section{Model and Preliminaries} \label{sec:prelim}

\subsection{Interdependent Valuations and Truthful Mechanisms} \label{sec:prelim-model}

We consider a single-item auction with $n$ bidders with interdependent valuations. 
(In later sections, we extend our work to $k$ identical items and unit-demand bidders). 
Every bidder $i \in [n]$ receives a private signal $\sig_i \in S_i$, where $S_i$ denotes the signal space of bidder $i$. 
We denote by $\sigspace = S_1\times \ldots \times S_n$ the joint signal space of the bidders, and by $\sigs = (\sig_1, \ldots, \sig_n) \in \sigspace$ a signal profile.
As is standard, we denote by $\sigs_{-i}=(\sig_1, \ldots, \sig_{i-1},\sig_{i+1},\ldots, \sig_n)$ the signal profile of all bidders other than bidder $i$.

In addition, every bidder $i$ has a private valuation function $\val_i: \sigspace\rightarrow \mathbb R_+$, which maps a signal profile into a non-negative real number, which is bidder $i$'s value for the item.
We denote by $V_i \subseteq \mathbb R_+^{\sigspace}$ the valuation space of bidder $i$, and by $\mathbf{V} = V_1\times\ldots\times V_n$ the joint valuation space of all bidders. 
A vector $\vals = (\val_1, \ldots, \val_n) \in \mathbf{V}$ denotes a valuation profile. 

A mechanism is defined by a pair $(\allocs, \prices)$ of an allocation rule $\allocs: \sigspace\times\mathbf{V}\rightarrow [0,1]^n$ and a payment rule $\prices: \sigspace\times\mathbf{V}\rightarrow 
\mathbb{R}_+^n$, which receive bidder reports about their signals and valuations, and return an allocation and a payment for each bidder.  
$\alloc_i(\sigs,\vals)$ and $\price_i(\sigs,\vals)$ denote bidder $i$'s allocation probability and payment for reported signals and valuations $\sigs,\vals$, respectively.

Unless specified otherwise, we access bidder valuations via value queries; namely, given a signal profile $\sigs$, bidder $i$'s value oracle $\val_i$ returns $\val_i(\sigs)$. 
A mechanism is said to be polynomial if it makes a polynomial number of value queries.

A mechanism $(\allocs,\prices)$ is said to be {\em truthful} if it is an ex-post Nash equilibrium for the bidders to truthfully report their private information (signals and valuations).
In our query access model, truthfulness means that it is in every bidder's best interest to answer every query truthfully, given that other bidders do the same.

\begin{definition}[EPIC-IR]
A mechanism $(\mathbf{x}, \mathbf{p})$ is 
{\em ex-post incentive compatible (IC)} if for every $i\in[n],\sigs\in \sigspace,  \vals \in \mathbf{V},  \sig_i'\in S_i, \val_i'\in V_i$
\begin{equation}
x_{i}(\sigs, \vals) \cdot v_{i}(\sigs) - p_i(\sigs, \vals) \geq x_{i}(\sigs_{-i},\sig_i', \vals_{-i}, \val'_i) \cdot v_{i}(\sigs) - p_i(\sigs_{-i},\sig_i', \vals_{-i}, \val'_i).
\label{eq:IC}
\end{equation}
It is {\em ex-post individually rational (IR)} if for every $i\in[n]$, $\sigs\in\sigspace$, and $\vals\in\mathbf{V}$
\begin{equation}
x_{i}(\sigs, \vals) \cdot v_{i}(\sigs) - p_i(\sigs, \vals) \geq 0
\label{eq:IR}
\end{equation}
It is EPIC-IR if it is both ex-post IC and ex-post IR.
An allocation $\allocs$ is {\em EPIC-IR} implementable if there exists a payment rule $\prices$ such that the pair $(\allocs,\prices)$ is EPIC-IR. 
\end{definition}

{It is well known that even when the valuation functions are public, this is the strongest possible solution concept when dealing with interdependent valuations.\footnote{Dominant strategy incentive-compatibility does not make sense, as a bidder $i$ might not be willing to win if other bidders over-bid, which causes the winner to over-pay and incur a negative utility.}}

\citet{EdenGZ22} give a sufficient condition for an allocation rule $\allocs$ to be EPIC-IR implementable.

\begin{proposition}[\citet{EdenGZ22}]\label{lem:pricing}
An allocation rule $\allocs$ is EPIC-IR implementable if for every bidder $i$, $\alloc_{i}$ depends only on $\sigs_{-i}, \vals_{-i}$ and $\val_{i}(\sigs)$, and is non-decreasing in $\val_{i}(\sigs)$. 

For an (EPIC-IR) implementable $\allocs$, the corresponding payment rule $\prices$ is given by:
\begin{equation}
p_{i}(\sigs, \vals) := x_i(\sigs_{-i}, \vals_{-i}, \val_i(\sigs))\cdot \val_i(\sigs) - \int_0^{\val_i(\sigs)} x_i(\sigs_{-i}, \vals_{-i}, t) \, \mathrm dt.
\label{eq:pricing}
\end{equation}
\end{proposition}
That is, bidder $i$'s allocation may depend on all other bidders' signals and valuation functions, and it can only depend on bidder $i$'s signal $\sig_i$ or valuation function $v_{i}$ through the numerical value $v_i(\sigs)$. {\citet{EdenGZ22} show that this condition is almost necessary in order to be EPIC-IR implementable.\footnote{{The necessary conditions for EPIC-IR implementablity are (i) $x_i$ is monotone in $v_i(\sigs)$, and (ii) for a given $\sigs_{-i}$, the set of signals $s_i,s'_i$ and valuation functions $v_i,v'_i$ such that $v_i(s_i,\sigs_{-i}) = v'_i(s'_i,\sigs_{-i})$ and $x_i(v_i,s_i,\sigs_{-i}) \neq x_i(v'_i,s'_i,\sigs_{-i})$ has measure $0$.}}} 

For the purpose of tie-breaking, we introduce the following notation.
\begin{definition}[Lexicographic tie-breaking]
    Let $a_i,b_j\in\mathbb{R}^+$, where $a_i$ is associated with bidder $i$ and $b_j$ is associated with bidder $j$. Given an ordering of the bidders $\pi=(\pi(1),\ldots, 
    \pi(n))$, we say $a_i>_\pi b_j$ if either $a_i>b_j$, or $a_i=b_j$ and $\pi(i)>\pi(j)$. \label{def:lex-bigger}
\end{definition}
As an example, {consider two quantities $a_1,b_2$ and the identity permutation $\pi(i) = i$.
\begin{itemize}
    \item If $a_1 =1$ and $b_2=0$, then $a_1 >_\pi b_2$.
    \item If $a_1 =0$ and $b_2 = 0$, then $b_2 >_\pi a_1$.
\end{itemize}}

\subsection{Properties of Valuation over Signals} \label{sec:prelim-defs}

In this section we introduce several properties of valuation functions over signals.
Recall that a valuation function over signals is a function $v: S_1\times\ldots\times S_n \rightarrow \mathbb R_+$, which assigns a (non-negative) real value to every vector of bidder signals.

Assume that signal spaces are totally ordered (e.g., $S_i \subseteq\mathbb R$ for all $i$). 
Denote by $\sigs \succeq \mathbf{t}$ that $\sigs$ is coordinate-wise greater than or equal to $\mathbf{t}$.
We first present the definition of a monotone valuation function.

\begin{definition}[Monotone]
A valuation $v$ over signals is \emph{monotone} if $v(\sigs) \geq v(\mathbf{t})$ for all $\sigs \succeq \mathbf{t}$.
\end{definition}

Note that, {while prior work in the interdependent values literature often assumes valuation functions to be monotone over signals,}
we also consider non-monotone valuations (see \Cref{lem:sos is sb}).

We next present the definition of submodularity over signals, defined by \citet{EdenFFGK19}.

\begin{definition} [\sos]
A valuation function $v$ is \emph{submodular over signals (\sos)} if for every $i\in [n]$ and every $\sigs\succeq\mathbf{t}$, it holds that $v(s_i, \sigs_{-i}) - v(t_i, \sigs_{-i})\leq v(s_i, \mathbf{t}_{-i}) - v(t_i, \mathbf{t}_{-i})$.
\end{definition}

Note that, when signals are binary (i.e., $S_i=\{0,1\}$ for all $i$), \sos coincides with the classic notion of submodular set functions. 

For the next properties, we use the notation $\low{}{i} := \underset{o_i\in S_i}{\inf}\val(o_i,\sigs_{-i})$ to denote the lowest value {of a valuation function $v$} over all bidder $i$'s signals, for a given signal profile $\sig_{-i}$ of all bidders other than bidder $i$. {We sometimes refer to $\low{}{i}$ as a lower-estimate of $v$.}
Note that if $\val$ is monotone, then $\low{}{i} := \val(0,\sigs_{-i})$ (where we normalize the lowest signal in $S_i$ to be $0$).

For our results, we use the notion of $d$-\selfbounding valuations, defined as follows.

\begin{definition}[\selfbounding and $d$-\selfbounding]
\label{def:self-bounding-valuation}
A valuation function $v$ is \selfbounding over signals if for every $\sigs\in \sigspace$, 
\begin{equation}
\sum_{i=1}^n (\val(\sigs) - \low{}{i}) \leq \val(s).
\label{eq:self-bounding}
\end{equation}
Similarly, a valuation function $v$ is $d$-\selfbounding over signals, for some parameter $d\in [n]$, if for every $\sigs\in \sigspace$,
$\sum_{i=1}^n (\val(\sigs) - \low{}{i}) \leq  d\cdot \val(s)$.
\end{definition}

For example, any function of the form $v(\sigs) = \sum_{i=1}^n f_i(s_i)$ is \selfbounding. Another example of \selfbounding functions are monotone \sos functions. In fact, any \sos function is 2-\selfbounding as shown in the following proposition.

\begin{restatable}{proposition}{sossb}
\label{lem:sos is sb}
Every monotone \sos valuation function is \selfbounding.
Moreover, every (possibly non-monotone) \sos valuation function is $2$-\selfbounding.
\end{restatable}

\begin{proof}
Consider an \sos function $v: \sigspace \rightarrow \mathbb R_+$. Take $s,o \in \sigspace$ and partition $[n] = A \sqcup B$, with $A = \{a_1, \dots, a_k\}$ and $B = \{b_1, \dots, b_\ell\}$ such that $o_A \preceq s_A$ and $s_B \preceq o_B$.
\begin{itemize}
\item For all $1 \leq i \leq k$ let $A_i = \{a_1, \dots, a_i\}$ and $s' = (o_{A_i}, s_{-A_i})$. We have $s' \preceq s$, and thus, by \sos for all $i\in\{2,\dots,k\}$ we have,
$$ \val(s)-\val(o_{a_i},s_{-a_i}) \leq \val(o_{A_{i-1}},s_{-A_{i-1}}) - \val(o_{A_i}, s_{-A_i}).$$
The inequality with $i=1$ is trivial. Therefore, summing over all $1 \leq i \leq k$, we obtain
$$ \sum_{a\in A}(\val(s) - \val(o_a, s_{-a})) \leq \val(s) - \val(o_A,s_B) \leq \val(s).$$
\item For all $1 \leq i \leq \ell$ let $B_i = \{b_1, \dots, b_i\}$ and $s' = (o_{B_i}, s_{-B_i})$. We have $s \preceq s'$, and thus, by \sos for all $i\in\{2,\dots,\ell\}$ we have,
$$\val(s)-\val(o_{b_i},s_{-b_i}) \leq \val(o_{B_{i-1}},s_{-B_{i-1}}) - \val(o_{B_i}, s_{-B_i}).$$
The inequality with $i=1$ is trivial. Thus, summing over all $2 \leq i \leq \ell$, we obtain
$$  \sum_{b\in B}(\val(s)-\val(o_b, s_{-b})) \leq \val(s) - \val(o_B,s_A) \leq \val(s).$$
\end{itemize}
Recall that $\low{}{i} := \inf_{o_i} \val(o_i, \sigs_{-i})$ for all $i$. If $v$ is non-decreasing, we set $A = [n]$ and $B = \emptyset$, and we use the first inequality. In the general case, we define the sets $A$ and $B$ according to the side of each infimum, and we sum the two equations, concluding the proof.
\end{proof}

A stricter notion than \selfbounding is that of a \critical (or $d$-\critical) valuation, defined as follows. 

\begin{definition}[$d$-\critical]
\label{def:bounded-valuation}
A valuation $v$ is $d$-\emph{critical} over signals for some parameter $d\in [n]$ if  for every $\sigs\in \sigspace$, the number of bidders $i$ such that $\val(\sigs) > \low{}{i}$ is at most $d$.
It is said to be \emph{critical} if this number is at most $1$.
\end{definition}

{We note that even the case of $d=1$ captures interesting scenarios (previously studied in the literature), such as the case where $v_i(\sigs) = \max_j v_{i,j}(s_j)$ for all $i$}.

\begin{proposition}
Every $d$-\critical valuation function is $d$-\selfbounding.
\end{proposition}

\begin{proof}
    Given a $d$-\critical function $v$, observe that we can bound each term $\val(\sigs)-\low{}{i}$ by $\val(\sigs)$ if $\low{}{i} < \val(\sigs)$, and by $0$ otherwise. Summing over all $i$ gives that $v$ is $d$-\selfbounding.
\end{proof}

\citet{EdenGZ22} studied valuations {termed} {\em $d$-bounded dependency valuations} which depend on at most $d$ signals. Obviously $d$-bounded dependency valuations are also $d$-\critical, {which in turn are $d$-\selfbounding}. \citet{EdenGZ22}
show that no EPIC-IR mechanism can give a better than $O(d)$-approximation for $d$-bounded dependency valuations. Thus, their result implies the following.

\begin{proposition}[Follows from Proposition 4.1 in \citet{EdenGZ22}]
\label{cor:d_sb_impossibility}
        For every $d$, no EPIC-IR mechanism can give better than $(d+1)$-approximation for $d$-{\selfbounding} valuations, even if the valuations are public.
\end{proposition}


\section{Main Ideas of Our Techniques}\label{sec:ideas-techniques}

\paragraph{Starting point: $d$-\critical valuations.}
We begin with $d$-\critical valuations. 
When $d$ is a constant, the mechanism devised by \citet{EdenGZ22} for $d$-bounded dependency valuations gives a constant-factor approximation to the optimal social welfare{ for this restricted class of valuations}. {{We note that this mechanism is substantially different from their $O(\log^2 n)$-approximation for \sos valuations.}}
For simplicity of presentation, our description below refers to $1$-\critical valuations.
{The entire discussion extends easily to $d$-\critical valuations.}

The mechanism can be described as follows: 
For every bidder $i$, compare $i$'s value $\val_i(\sigs)$ to the other bidders' $j\neq i$ values under the worst-possible signal of bidder $i$, namely $\low{j}{i} = \inf _{o_i} v_j(o_i, \sigs_{-i})$. 
Then, if $\val_i(\sigs) > \low{j}{i}$ for every other bidder $j\neq i$, bidder $i$ is said to be a ``\textit{candidate}'' (to be allocated to), which means that the item is allocated to bidder $i$ with probability $x_i=1/2$.

As they show, this mechanism is EPIC-IR (as for every $\vals_{-i},\sigs_{-i}$, $i$'s allocation probability is monotone in $\val_i(\sigs)$), and gives a $1/2$-approximation (since the highest-valued bidder wins with probability at least $1/2$).
Moreover, if there is a unique highest-valued bidder, call it $i^\star$, then this mechanism is feasible, since the only bidder that has a chance to receive a non-zero allocation probability aside from $i^\star$ is the one bidder whose signal can decrease $i^\star$'s value at signal profile $\sigs$ (bidder $i'$ in \Cref{fig:bounded-independence}). Since there is at most one such bidder, and since both get allocation probability $1/2$, the mechanism is feasible.
Otherwise (if there is no unique highest value), it can be handled by using a fixed lexicographic tie-breaking between identical values {when deciding if a bidder is a candidate}.

\begin{figure}[ht!]
    \centering
    \begin{tikzpicture}[thick,>=stealth]
    \draw[very thick,->] (0,0) -- (0,3);
    \draw[dashed] (0,2.5) -- (10,2.5);
    \draw[ForestGreen][dashed] (0,1) -- (10,1);
    \node[anchor=east] at (0,1.5)
        {$\mathrm{Value}$};
    \fill[NavyBlue] (2.5,2.5) circle (1mm);
    \node[anchor=south] at (2.5,2.5) {\footnotesize$v_{i^\star}(\sigs)$};
    \fill[ForestGreen] (2.5-.1,0.94) rectangle (2.5+.1,1.06);
    \node[anchor=north] at (2.5,1) {\footnotesize$\low{i^\star}{i'}$};
    \draw[densely dotted] (2.5,1.07) -- (2.5,2.4);
    \fill[ForestGreen] (4,1.8) circle (1mm);
    \node[anchor=south] at (4,1.8) {\footnotesize$v_{i'}(\sigs)$};
    \fill (6,1.3) circle (1mm);
    \node[anchor=south] at (6,1.3) {\footnotesize$v_{\ell}(\sigs)$};
    \fill (8,0.5) circle (1mm);
    \node[anchor=south] at (8,0.5) {\footnotesize$v_{k}(\sigs)$};
    \node at (8,2.9) {\footnotesize$\low{i^\star}{j} = \val_{i^\star}(\sigs)\quad(\forall j\neq i')$};
    \end{tikzpicture}
    \caption{Illustration of the mechanism in~\citet{EdenGZ22} for   bounded-dependency valuations. There are at most two  candidates:~the highest-valued bidder $i^\star$ (breaking ties {consistently}),
    and the only bidder $i'$ 
    {whose} signal affects the value of $i^\star$. No other bidder (e.g., $k$ or $\ell$) can be a candidate.}
\label{fig:bounded-independence}
\end{figure}
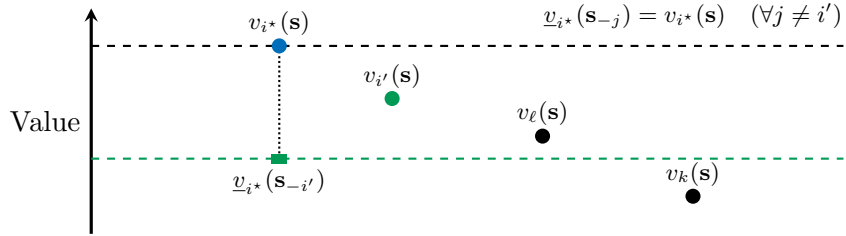

This idea can be extended to a $(d+1)$-approximation mechanism for $d$-\critical valuations. 
In this case, we set $x_i = 1/(d+1)$ instead.

We next describe how to build upon this idea for the general case of \sos valuations.

\paragraph{\sos valuations.}
Unfortunately,  for \sos valuations, the previous observation {has no bearing}, as for any given signal profile, $\val_i(\sigs)$ can be affected by an arbitrary number of bidders{, leading to an $\Omega(n)$-approximation}. 
Here, \Cref{lem:sos is sb} comes to our aid. Namely, since \sos valuations are $1$-\selfbounding, there is at most one bidder that can decrease $i$'s value by a factor larger than 2.

\paragraph{First attempt: discretized values.}
As a na\"ive first attempt, consider simply discretizing the valuation space into powers of 2, and rounding down every value to the nearest power of two. That is, {for a valuation function $v$, we define the discretized value} $\tilde{v}(\sigs) = 2^{\lfloor \log \val (\sigs)\rfloor}$ for all signal profiles $\sigs$. If it happens to be the case that the valuations are slightly below a power-of-two (i.e., $\val(\sigs) = 2^{\ell} - \varepsilon$), then they need to be lowered by a factor of at least $1/2$ in order {for $v(\sigs)$ and $\low{}{i}$ to be discretized to different powers-of-two.}
In this case, at most one bidder can decrease the value of a bidder from its true rounded-down value to the next rounded-down value, and so we could use the mechanism for $1$-\critical valuations and lose another factor of $2$ due to the discretization. 
Unfortunately, the discretized value may, in general, be affected by all bidders, even for \sos valuations. Thus, even the discretized valuations can be $n$-critical. 
This is demonstrated in the following example. 

\begin{example} \label{ex:nonrandomrounding}
Consider the $1$-\selfbounding valuation function $v(\sigs) = \frac{2(1+\epsilon)}{n} \cdot \sum_i s_i$, where $s_i=1$ for all $i$ and $\epsilon < 1/(n-1)$.~Its discretization is $\tilde{v}(\sigs)=2$.~For every $i$, ${\val}(0_i,\sigs_{-i}) < 2$, thus $\tilde{\val}(0_i,\sigs_{-i}) =1$, and $n$ different bidders can decrease $i$'s rounded-down value to the next power-of-two. 
\end{example}

Indeed, the valuation $v(\sigs)$ in \Cref{ex:nonrandomrounding} is monotone \sos (and $1$-\selfbounding), but the na\"ive discretization results in all $n$ bidders being able to decrease the discretized valuation at the true signal profile of $\sigs = (1,1,\cdots,1)$.

\paragraph{Second attempt: randomized discretization.}
Our next attempt is using \emph{randomized discretization}. 
Concretely, rather than rounding down to the nearest $2^\ell$ {from below} for $\ell \in \mathbb{Z}$, we draw $r \sim U[0,1)$ and round down to the nearest $2^{\ell+r}$ {from below}.
{We show that}
using this random discretization, \emph{in expectation}, a constant number of bidders can decrease the (randomly) discretized valuation.  
 
However, even with our randomized discretization, another problem may arise: a fixed tie-breaking order can lead to too many bidders being candidates, which results in a feasibility problem. 
This is demonstrated in the following example.

\begin{example} \label{ex:tiebreaking} For each bidder $i\geq\sqrt{n}$, let $\underline{i} = i - \sqrt{n} +1$, we have $$\val_i(\sigs) = 2^{(n-i)/n}\cdot\sum_{j={\underline{i}}}^{i}s_j/\sqrt{n}.$$
For $\sigs =(1,1,\ldots,1)$, the values of the all the bidders $i\ge \sqrt{n}$ are almost evenly distributed between $1$ and $2$; specifically, $\val_i(1,\ldots,1)=2^{(n-i)/n}$. 
Each bidder $j$ can affect values which are slightly smaller than their own, that is $\low{i}{j} = \val_i(\sigs)\cdot (1-1/\sqrt{n})$ for all {$i \in [j, j+\sqrt{n} -1]$}, and $\low{i}{j} = \val_i(\sigs)$ for every other $i$. We illustrate this visually in \Cref{fig:tie-breaking-example}. 
\end{example}
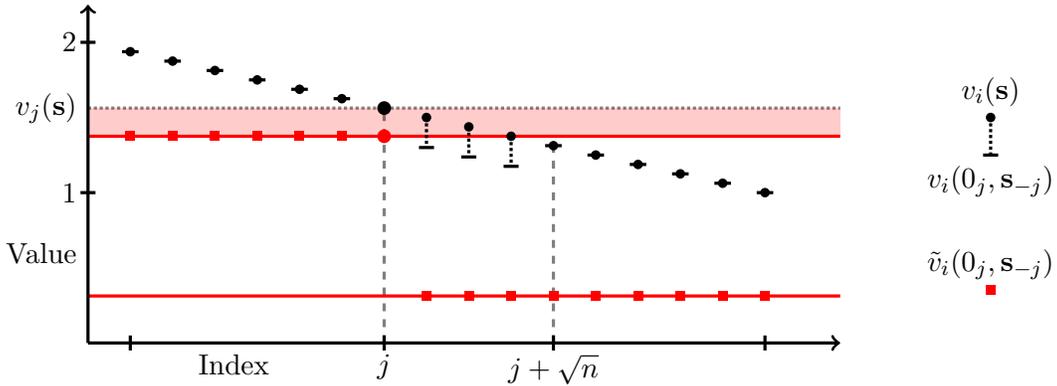
\begin{figure}[ht]
    \begin{center}
    \begin{tikzpicture}[very thick]
    \def\i{7}
    \def\o{-0.5}
    \draw[densely dotted, black!50!white] (0,3.5-2*\i/16) -- (10,3.5-2*\i/16);
    \draw[dashed, black!50!white] (9*\i/16,\o) -- (9*\i/16,3.5-2*\i/16);
    \draw[dashed, black!50!white] (9*\i/16+9/4,\o) -- (9*\i/16+9/4,3.5-2*\i/16-2/4);
    \fill[red, opacity=0.2] (0,3.5-2*\i/16) rectangle (10,3.5-2*\i/16-2*3/16);
    \draw[red] (0,3.5-2*\i/16-2*3/16) -- (10,3.5-2*\i/16-2*3/16);
    \draw[red] (0,3.5-2*\i/16-2-2/4) -- (10,3.5-2*\i/16-2-2/4);
    \node[anchor=east] at (0,3.5-2*\i/16) {$\val_j(\sigs)$};
    \node[anchor=east] at (0,3.5) {$2$};
    \node[anchor=east] at (0,1.5) {$1$};
    \node[anchor=east] at (0,0.7) {$\mathrm{Value}$};
    \draw (9*\i/16,-.1+\o) -- (9*\i/16,.1+\o);
    \draw (9*\i/16+9/4,-.1+\o) -- (9*\i/16+9/4,.1+\o);
    \draw (9/16,-.1+\o) -- (9/16,.1+\o);
    \draw (9,-.1+\o) -- (9,.1+\o);
    \node[anchor=north] at (9*\i/16, \o) {$j$};
    \node[anchor=north] at (9*\i/16+9/4, \o) {$j+\sqrt{n}$};
    \node[anchor=north] at (9*\i/16-2, \o) {$\mathrm{Index}$};
    \draw[->] (0,\o) -- (10,\o);
    \draw[->] (0,\o) -- (0,4);
    \draw (-.1,1.5) -- (.1,1.5);
    \draw (-.1,3.5) -- (.1,3.5);
     \fill (9*\i/16,3.5-2*\i/16) circle (2.6pt);
     \fill[red] (9*\i/16,3.5-20/16) circle (2.6pt);
    \foreach \i in {1,2,...,6} {
        \draw (9*\i/16-.11,3.5-2*\i/16) -- (9*\i/16+.11,3.5-2*\i/16);
        \fill[red] (9*\i/16-.06,3.44-20/16) rectangle (9*\i/16+.06,3.56-20/16);
        \fill (9*\i/16,3.5-2*\i/16) circle (1.8pt);
    }
    \foreach \i in {8,...,16} {
        \fill[red] (9*\i/16-.06,3.44-22/16-2) rectangle (9*\i/16+.06,3.56-22/16-2);
    }
    \foreach \i in {11,...,16} {
        \draw (9*\i/16-.11,3.5-2*\i/16) -- (9*\i/16+.11,3.5-2*\i/16);
        \fill (9*\i/16,3.5-2*\i/16) circle (1.8pt);
    }
    \foreach \i in {8,...,10}  {
        \draw[densely dotted] (9*\i/16,3.5-2*\i/16) -- (9*\i/16,3.5-2*\i/16-.4);
        \draw (9*\i/16-.1,3.5-2*\i/16-.4) -- (9*\i/16+.1,3.5-2*\i/16-.4);
        \fill (9*\i/16,3.5-2*\i/16) circle (1.8pt);
    }
    \fill (12,2.5) circle (1.8pt);
    \draw[densely dotted] (12,2.5) -- (12,2);
    \draw (11.9,2) -- (12.1,2);
    \fill[red] (12-.06,.14) rectangle (12+.06,.26);
    \node[anchor=south] at (12,2.5) {$\val_i(\sigs)$};
    \node[anchor=north] at (12,2) {${\val}_i(0_j,\sigs_{-j})$};
    \node[anchor=south] at (12,.2) {$\tilde{\val}_i(0_j,\sigs_{-j})$};
    \end{tikzpicture}
    \end{center}
    \caption{{Illustration of \Cref{ex:tiebreaking} showing that the expected number of candidates can be $\sqrt{n}$ using a fixed tie-breaking rule.}}
    \label{fig:tie-breaking-example}
\end{figure}

In \Cref{ex:tiebreaking}, the true values of the bidders decrease with their index/name. Suppose we break ties according to the identity permutation $\pi$ (i.e., $\pi(i)=i$ for all $i$).
This means that by \Cref{def:lex-bigger}, in the above example, if the discretized values are the same for two bidders, we break ties in favor of the bidder with the lower true value. We claim that each bidder is a candidate with probability $1/\sqrt{n}$. This is because, bidder $j$ is a candidate when $\val_{j}(\sigs)$ and $\low{i}{j}$ for $i<j$ are rounded to the same discretization point, and  all $\low{i}{j}$ for $i>j$ are rounded to a lower discretization point {(the discretization points are illustrated by the two red lines in~\Cref{fig:tie-breaking-example})}. {This corresponds to the event that the corresponding random discretization point falls between $2^{(n-j)/n}$ and $2^{(n-j-\sqrt{n})/n}$ (the shaded red area illustrated in~\Cref{fig:tie-breaking-example}); i.e., when we draw $r\sim U[0,1)$ it falls in $(1 - j/n -1/\sqrt{n},~1-j/n]$, an interval of length $1/\sqrt{n}$}. Hence, the expected number of bidders who are candidates is $\sqrt{n}$. {So if we directly attempt to use the above mechanism {that is designed for the case} where at most $d$ bidders can decrease one's value, we need to normalize the allocation probabilities $x_i$ by $\sqrt{n}$ {in order to preserve feasibility}, 
which in turn leads to a $\sqrt{n}$-factor loss in the approximation ratio}. 

\paragraph{Final solution: randomized tie-breaking.}
To handle this, we turn to a random tie-breaking rule.
Specifically, we choose a tie-breaking rule by picking a tie-breaking order (permutation) uniformly at random. The intuition here is reminiscent of the secretary problem where we bound the probability of prematurely accepting a sub-optimal element by relying on the second-best appearing earlier in the sample. Here, we observe that if we break ties in favor of a bidder $i < j$ (i.e., $i$ {ranks above} $j$ in the random tie-breaking order) then bidder $j$ has no chance of being a candidate. {Following this intuition, we turn to using a random permutation in order to break ties. However, note that there is still a (worst-case) order with a large number of candidates as shown above in \Cref{ex:tiebreaking}. Therefore, we set the allocation probabilities proportional to the \emph{probability that bidder} $j$ \emph{is a candidate}, instead of giving a fixed allocation probability to all the bidders who are candidates under a given randomization.}

Our Randomized Candidate Filtering~(\ref{alg:SB})~Mechanism  combines the concepts of randomized discretization and random tie-breaking as follows. 
It first considers the rounded-down valuations to randomly selected discretization points. It marks bidder $i$ as a ``candidate'' to be allocated if their discretized value is larger than all the other bidders' discretized lower-estimates, breaking ties according to a random permutation. 
Mechanism~\ref{alg:SB} then  allocates the item to bidder $i$ according to the probability that $i$ will be a candidate when choosing a random $r\sim U[0,1]$ and a uniformly random permutation $\pi$. 

In \Cref{lem:dsb_truthful}, we show that the mechanism is EPIC-IR since for a fixed $\sigs_{-i}, \vals_{-i}$ increasing $v_i(\s)$ also increases the probability of being a candidate. This mechanism is not yet feasible, as the expected number of candidates can 
exceed 1.
Therefore, we normalize the allocation probability by (an upper bound on) the expected number of candidates. \emph{Bounding the expected number of candidates is our main technical challenge.}

\section[O(d)-Approximation for d-self-bounding]{$O(d)$-Approximation for $d$-\selfbounding} \label{sec:d-sb-mech}
    
In this section, we prove our main result. Building upon the intuition of the previous sections, we define our mechanism for instances with $d$-\selfbounding valuations. We prove truthfulness and show the desired  $O(d)$-approximation guarantee which is optimal (up to constants). 
As \sos valuations are $2$-\selfbounding, this implies a constant-factor approximation EPIC-IR mechanism for \sos valuations, answering the open question raised by~\citet{EdenGZ22} in the affirmative.

\begin{theorem}\label{thm:sos-constant}
    There exists an EPIC-IR mechanism that obtains a tight $\Theta(d)$-approximation to the optimal welfare for any instance with $d$-\selfbounding valuations. Specifically, 
    \begin{itemize}
        \item There exists a $5.55$-approximation mechanism for monotone $\sos$ valuations.
        \item There exists a $8.32$-approximation mechanism for (non-monotone) $\sos$ valuations. 
        \item The mechanism can be made oblivious to $d$ by losing another factor of $2$ in the approximation.
    \end{itemize}
    Finally, the allocation and payments can be computed in polynomial time.
\end{theorem}

Our mechanism, the Randomized Candidate Filtering (\ref{alg:SB}) Mechanism, operates as follows. It rounds down valuations to randomly selected thresholds around powers of $2$, and marks a bidder as \textit{candidate} by setting $c_i=1$ if their rounded-down value is lexicographically larger (\Cref{def:lex-bigger}) than all other bidders' rounded-down lower-estimates. {Tie-breaking is done using a randomly drawn permutation}. Our mechanism assigns each bidder an allocation probability which is {proportional to} the probability this bidder will be a candidate (that is, sets $x_i=\E[c_i]/\eta$, {for a normalization factor $\eta$}). The probability is taken over the random rounding and the random lexicographic tie-breaking. 

{The desiderata of the mechanism are: (i) truthfulness, (ii) constant approximation to the optimal social welfare, and (iii) allocation feasibility.}

The mechanism is truthful since as a bidder's value increases,  $c_i$ can only increase (\Cref{lem:dsb_truthful}).

{To show that the mechanism achieves a good approximation, we show that for every random coin toss of the algorithm, there is always a nearly-optimal candidate (\Cref{lem:dsb_approx}).}

{The main technical challenge is the third desideratum, namely feasibility. 
As there can be more than one candidate for a given random seed, the mechanism need not be feasible. 
The main technical challenge is indeed showing that for $d$-\selfbounding valuations, the expected number of candidates is $O(d)$; this is established in \Cref{lem:dsb_feasibility}.}
Therefore, by normalizing by a factor $O(d)$, we retain feasibility and get an $O(d)$-approximation algorithm. As \sos valuations are $1$-\selfbounding, this implies a constant-factor approximation EPIC-IR mechanism for \sos valuations, answering the open question raised by~\citet{EdenGZ22} in the affirmative.  Our mechanism follows.

\begin{algorithm}[h]
    \begin{enumerate}
\item Elicit reported signals $\sigs[\hat] = \{\sig[\hat]_i \in S_i\}_{i\in [n]}$ and values $\vals[\hat] = \{\val[\hat]_i: \sigspace\rightarrow \mathbb R_+\}_{i\in [n]}$.
\item Let $\eta \ge 1$ be a normalization parameter to be set later. 

\item For each bidder $i$, define \[x_i = \frac{\E_{r,\pi}[c_i]}{\eta},\]  where  $r$ is uniformly distributed on $[0,1)$, $\pi$ is a uniformly random permutation, and  $c_i$ is an indicator variable defined as follows:
\[c_i =
\begin{cases}
1, & \text{if } f_r(\val[\hat]_i(\sigs[\hat])) >_{\pi} f_r(\low[\hat]{j}{i}) \text{ for all } j\neq i \text{ (recall \Cref{def:lex-bigger})}\\
0, & \text{otherwise}
\end{cases}\]
where $f_r(w) := {2}^{r+k}$ such that ${2}^{r+k} \leq w < {2}^{r+k+1}$ for all $w$.
\item Allocate the item to bidder $i$ with probability $x_i$ for all $i\in [n]$.
\item Charge prices using \Cref{eq:pricing}.    
\end{enumerate}
\renewcommand{\thealgorithm}{RCF}
\caption{\textbf{Randomized Candidate Filtering (RCF) Mechanism}.}
\label{alg:SB}
\end{algorithm}

We begin by showing the mechanism is truthful.
\begin{lemma} \label{lem:dsb_truthful}
    The \ref{alg:SB} Mechanism is EPIC-IR. 
\end{lemma} 
\begin{proof}
    Fix bidder $i$ and reported signals and valuations $\sigs[\hat]_{-i}$, $\vals[\hat]_{-i}$ {of the other bidders}.  We show that for every choice of $r$ and $\pi$, $c_i$ is monotone in $\hat{v}_i(\sigs[\hat])$. This immediately implies that $x_i$ is monotone as well which implies the mechanism can be implemented in an EPIC-IR manner by \Cref{lem:pricing}.  Fix $r$ and $\pi$. If we increase $\val[\hat]_i(\sigs[\hat])$, then for every $j$, we can only get  $f_r(\val[\hat]_i(\sigs[\hat])) >_{\pi} f_r(\low[\hat]{j}{i})$ to be satisfied if it wasn't satisfied before. This is since the left hand side of the inequality increases while {the} right hand side is not affected. Therefore, $c_i$ is monotone in $\val[\hat]_i(\sigs[\hat])$ which proves the {lemma}.
\end{proof}

As the mechanism is truthful, from now on we assume bidders bid their true valuations and signals, and write $\sigs$ and $\vals$ instead of $\sigs[\hat]$ and $\vals[\hat]$.

We next show that the mechanism obtains near-optimal welfare.

\begin{lemma}
    The~\ref{alg:SB} Mechanism obtains an ${(\eta\cdot 2\ln 2)}$-approximation to the optimal welfare. \label{lem:dsb_approx} 
\end{lemma} 
\begin{proof}
   Fix valuations and signals $\vals, \sigs$ and consider 
   the random choice of $r$ and $\sigma$. Consider the bidder $i^\star$ such that $f_r(v_{i^\star}(\sigs)) >_\pi f_r(v_j(\sigs))$ for every $j$. For bidder $i^\star$ it must be the case that $c_{i^\star}=1$ as  $f_r(v_j(\sigs))\geq f_r({\low{j}{i^\star}})$ for every $j$. Moreover, we have that 
   \begin{eqnarray*}
       v_{i^\star}(\sigs) \ge f_r(v_{i^\star}(\sigs)) \ge  f_r(\max_i v_{i}(\sigs)).
   \end{eqnarray*}
   Therefore, for every $\vals,\sigs, r,\pi$,  we have that $\sum_i c_i \cdot v_i(\sigs)\geq v_{i^\star}(\sigs)\ge f_r(\max_i v_{i}(\sigs))$, and 
   \begin{eqnarray}
        \sum_i x_i(\v,\s) \cdot v_i(\s) & = & \sum_i \frac{\E_{r,\pi}[c_i]}{\eta} \cdot v_i(\s) \nonumber \\
            & = &  \frac{\E_{r,\pi}[\sum_i c_i\cdot v_i(\s)]}{\eta} \nonumber\\
            & \geq &  \E_{r}[f_r(\max_i v_{i}(\sigs))]/\eta.\label{eq:eff_lb}
   \end{eqnarray}
   To finish the proof, we show that for a positive real $v\in \mathbb R_+$, $\E_r[f_r(v)] \geq \frac{v}{2\ln 2}$.
   Let $v=2^{k+\alpha}$ for some $k\in \mathbb N$ and $\alpha\in [0,1]$, and let $r$ be the random number sampled in step (3) of \ref{alg:SB}. If $r$ is chosen such that $r\le \alpha$, then $v$ is rounded down to ${2}^{k+r}$. 
   On the converse, if $r > \alpha$, then $v$ is rounded down to ${2}^{k+r-1}$. Overall, 
   \begin{eqnarray}
       \E_r[f_r(v)]  & = &  \int_0^\alpha 2^{k+r}dr + \int_\alpha^1 2^{k+r-1}dr \nonumber\\
       & =&  \frac{2^{k+\alpha} - 2^k +  2^k-2^{k+\alpha-1}}{\ln 2} \nonumber \\
       & = &  \frac{2^{k+\alpha}}{2\ln 2}  \ =\  \frac{v}{2\ln 2}.\label{eq:round_v}
   \end{eqnarray}
   
   Combining Equations~\eqref{eq:eff_lb} and~\eqref{eq:round_v} gives the desired bound.
\end{proof}

\subsection[The~\ref{alg:SB} Mechanism is Feasible]{The~\ref{alg:SB} Mechanism is Feasible for $\eta=O(d)$}
In this section we show that the~\ref{alg:SB} mechanism is feasible when the valuations are $d$-\selfbounding\  when using a normalization factor $\eta=O(d)$. We first consider the setting where $d$ is known, and extend the result to the setting where $d$ is unknown in \Cref{sec:unknown-d}. In our proof, we use the following notation for convinience.

\begin{definition}
    For any $\alpha\ge 0,$ $\logdag(\alpha) = \max(0, \min(1, \log_2 \alpha))$.    
\end{definition}

The following property of $\logdag$ is used in our proofs.
\begin{lemma} \label{obs:logdag-prod}
    For any $\alpha,\beta\ge 0$, $\logdag(\alpha\cdot \beta)\le \logdag(\alpha)+\logdag(\beta).$ 
\end{lemma}
\begin{proof}
For any $x,y\in \mathbb R$, we have $\max(0, x+y) \leq \max(0, x) + \max(0, y)$ and $\min(1, x+y) \leq \min(1, x) + \min(1, y)$.
Therefore,
\begin{eqnarray*}
    \logdag(\alpha\cdot \beta) & =  &\max(0, \min(1, \log_2 \alpha\cdot \beta)) \\ 
    & = & \max(0, \min(1, \log_2 \alpha+ \log_2 \beta))  \\
    &\le & \max(0, \min(1, \log_2 \alpha) +\min(1, \log_2 \beta )) \\
    & \le & \max(0, \min(1, \log_2 \alpha)) + \max(0, \min(1, \log_2 \beta))\\
    & = & \logdag\alpha + \logdag\beta.
\end{eqnarray*}.
\end{proof}

The following lemma shows it is enough to normalize the allocation by a factor $\eta=O(d)$ in order to maintain feasibility.

\begin{lemma}\label{lem:dsb_feasibility}
    For every single-item auction with $d$-\selfbounding valuations, the~\ref{alg:SB} Mechanism with {$\eta = 2(d+1)$} is feasible. 
\end{lemma}

\begin{proof}
    We show that the expected number of candidates $\sum_{i=1}^n\E_{r,\pi}[c_i]$ is at most  ${2(d+1)}$. By the definition of $x_i$ this implies that $\sum_i x_i \le 1$ for $\eta = 2(d+1)$, thus proving feasibility. 
    
    First, we rename the bidders such that $\val_1(\sigs) \ge \val_2(\sigs) \ge \cdots \ge \val_n(\sigs)$. We set $k$ to be the number of bidders $i$ whose value are larger than $\val_1(\sigs)/2$, that is $k = \max\{i\,|\,\val_i(\sigs) > \val_1(\sigs)/2\}$. This parameter distinguishes the analysis for large valued bidders (numbered $1$ through $k$), and small valued bidders (numbered $k+1$ to $n$). By \Cref{lem:proba-candidate} we get the following bound on the probability of being a candidate,
    $$
    \forall i\in[n],\qquad
    \E_{r,\pi}[c_i] \leq
    \frac{1}{i(i+1)}+
    \frac{\logdag(2\val_i(\sigs)/\low{1}{i})}{k+1}
    +\sum_{j\in[k]}\frac{\logdag(\val_i(\sigs)/\low{j}{i})}{j(j+1)}.
    $$
    Next, using \Cref{obs:logdag-prod}, we write
    \begin{equation}\label{eq:ci_decomposition}
    \begin{array}{rcrcr}
    \displaystyle\E_{r,\pi}[c_i] \leq \frac{1}{i(i+1)}
    &+&
    \cellcolor{black!10}
    \displaystyle\frac{\logdag(\val_1(\sigs)/\low{1}{i})}{k+1}
    &+&
    \cellcolor{black!10}
    \displaystyle\frac{\logdag(2\val_i(\sigs)/\val_1(\sigs))}{k+1}
    \\
    &+&
    \cellcolor{black!10}
    \displaystyle\sum_{j\in[k]}\frac{\logdag(\val_j(\sigs)/\low{j}{i})}{j(j+1)}
    &+&
    \cellcolor{black!10}
    \displaystyle\sum_{j\in[k]}\frac{\logdag(\val_i(\sigs)/\val_j(\sigs))}{j(j+1)}
    \\[-.2cm]
    &&
    \underbrace{\hspace{4.1cm}}_{A_i}
    &&
    \underbrace{\hspace{3.5cm}}_{B_i}
    \end{array}
    \end{equation}
    Recall that we want to show that $$\sum_i\E_{r,\pi}[c_i]\ \le\ \sum_{i=1}^n\frac{1}{i(i+1)}+\sum_i A_i+\sum_i B_i \ \le\ 2(d+1).$$ First, observe that $\sum_{i=1}^n\frac{1}{i(i+1)} = \frac{n}{n+1} \leq 1$. To conclude the proof, we use \Cref{lem:feasibilityA} which shows that $\sum_i A_i\leq 2d$, and \Cref{lem:feasibilityB} which shows that $\sum_i B_i \leq 1$.
\end{proof}

In our proofs, we use the following technical lemma.

\begin{lemma}
\label{lem:rounded-threshold}
For any $a, b \in \mathbb R^+$, we have that
 \[\Pr_r[f_r(a) > f_r(b)] = \max(0, \min(1, \log_2(a/b))) = \logdag(a/b).\]
\end{lemma}

\begin{proof}
    First, note that if $a \leq b$, then $\Pr[f_r(a) > f_r(b)]=0$ and $\log_2(a/b) \leq 0$; and if $a\ge 2 b$, then $\Pr[f_r(a) > f_r(b)]=1$ and $\log_2(a/b) \geq 1$. Thus, we consider two cases:
    \begin{itemize}
        \item $a=2^{i+\alpha}$, $b=2^{i+\beta}$ for $\beta < \alpha < 1$. For this case, $\Pr[f_r(a) > f_r(b)]$ if $r\in (\beta,\alpha]$. This happens with probability $\alpha-\beta = \log_2 a - \log_2 b = \log_2(a/b)$. 
        \item $a=2^{i+\alpha}$, $b=2^{i-1+\beta}$, $\alpha \le \beta  < 1$. In this case, $f_r(a) > f_r(b)$ if (a) $r\le \alpha$, which implies  $f_r(a) = 2^{i+r} > 2^{i-1+r}=f_r(b)$, or (b) $r>\beta,$ which implies $f_r(a) = 2^{i-1+r} > 2^{i-2+r}=f_r(b)$.  These events are disjoint, and happen with probability $\alpha+1-\beta=\log_2 a - \log_2 b= \log_2(a/b)$.
    \end{itemize}
\end{proof}

{The following lemma establishes a useful property of $d$-\selfbounding functions.
Namely, that by the random discretization the expected number of bidders who can decrease some bidder $i$'s value to a lower discretization point is $O(d)$. 
Indeed, the left hand side of Equation~\eqref{eq:dsb-reduction} is the expected number of bidders who can decrease the discretized value at a signal profile to the next (lower) power of $2$. }

\begin{lemma}\label{lem:dsb}
{For any $d$-\selfbounding function $\val$, it holds that}
\begin{eqnarray}
\sum_{i=1}^n
\logdag\left(\frac{\val(\sigs)}{\low{}{i}}\right) \leq 2d. \label{eq:dsb-reduction}   
\end{eqnarray}
\end{lemma}

\begin{proof}
We define the function $\phi(x) = -\log_2(1-x)$ and we write
$$
\sum_{i=1}^n
\logdag\left(\frac{\val(\sigs)}{\low{}{i}}\right)
= \sum_{i=1}^n \phi(y_i)
\qquad\text{where }y_i := \min\left(1-\frac{\low{}{i}}{\val(\sigs)}, \frac{1}{2}\right)
$$
Because $\phi$ is convex, it lies below its chord between $\phi(0) = 0$ and $\phi(1/2) = 1$, thus
$$
\forall y_i\in[0,1/2],\qquad \phi(y_i) \leq \frac{y_i}{1/2} = 2 y_i.
$$
Using the $d$-\selfbounding property to derive the second inequality, we have that
$$\sum_{i=1}^n y_i \leq \sum_{i=1}^n\frac{\val(\sigs)-\low{}{i}}{\val(\sigs)} \leq d.$$
Therefore, summing over all $i$ we conclude that $$\sum_{i=1}^n
\logdag\left(\frac{\val(\sigs)}{\low{}{i}}\right) \ = \ \sum_{i=1}^n \phi(y_i)\ \leq\ \sum_{i=1}^n 2y_i\ \leq\ 2d.$$
\end{proof}

We first bound the $A_i$ terms.
\begin{lemma}\label{lem:feasibilityA}
    Given an instance with $d$-\selfbounding valuations, we have that
    $$\sum_{i=1}^n A_i \leq 2d,$$
    where $A_i$'s are defined in the proof of \Cref{lem:dsb_feasibility}.
\end{lemma}
\begin{proof}
    Recall definition from~\Cref{eq:ci_decomposition}
    \[A_i = \frac{\logdag(\val_1(\sigs)/\low{1}{i})}{k+1} + \sum_{j\in[k]}\frac{\logdag(\val_j(\sigs)/\low{j}{i})}{j(j+1)}.\] 
    Moreover, by \Cref{lem:dsb} we have $\sum_{i=1}^n\logdag(\val_j(\sigs)/\low{j}{i})\le 2d$ for all bidders $j$. Hence, summing $A_i$ over all $i$ and swapping the summation of $i$ and $j$ we get,
    \begin{align*}
        \sum_{i=1}^n A_i & = \sum_{i=1}^n\frac{\logdag(\val_1(\sigs)/\low{1}{i})}{k+1} + \sum_{j\in[k]}\sum_{i=1}^n\frac{\logdag(\val_j(\sigs)/\low{j}{i})}{j(j+1)} \\
        & \le \frac{2d}{k+1} + \sum_{j\in[k]}\frac{2d}{j(j+1)}\\
        & \le \frac{2d}{k+1} + 2d\cdot\frac{k}{k+1} \\
        & \le  2d.
    \end{align*}
\end{proof}

Next, we bound the $B_i$ terms.
\begin{lemma}\label{lem:feasibilityB}
    Given an instance, we have that
    $$\sum_{i=1}^n B_i \leq 1,$$
    where $B_i$'s are defined in the proof of \Cref{lem:dsb_feasibility}.
\end{lemma}
\begin{proof}
Recall that $k = \max\{i\,|\,\val_i(\sigs) > \val_1(\sigs)/2\}$. Therefore, for every $i>k$, $$\logdag(2\val_i(\sigs)/\val_1(\sigs))\le \logdag(1) = 0.$$ For $i\le k$, we first observe that
$$\log_2(2\val_i(\sigs)/\val_1(\sigs))\le 1.$$ Moreover for $i\le j\le k,$ we have. $$\log_2(\val_i(\s)/\val_j(\s))\le \log_2(\val_1(\s)/\val_1(\s)/2) \le 1,$$
hence we can replace $\logdag$ with $\log_2$ for these terms. Therefore,
\begin{align}
\sum_{i=1}^n B_i & =
\sum_{i=1}^n \left(\frac{\logdag(2\val_i(\sigs)/\val_1(\sigs))}{k+1}+
\sum_{j=1}^k
\frac{\logdag(\val_i(\sigs)/\val_j(\sigs))}{j(j+1)}\right)\notag\\
    &=
\sum_{i=1}^k
\frac{\log_2(2\val_i(\sigs)/\val_1(\sigs))}{k+1}+
\sum_{i=1}^k\sum_{j= i}^k
\frac{\log_2(\val_i(\sigs)/\val_j(\sigs))}{j(j+1)}\notag\\
&= \sum_{i=1}^k\frac{1+\log_2(\val_i(\sigs))-\log_2(\val_1(\sigs))}{k+1} + \sum_{i=1}^k\sum_{j= i}^k
\frac{\log_2(\val_i(\sigs))-\log_2(\val_j(\sigs))}{j(j+1)}.\label{eq:bi_bound}
\end{align}
We now bound the second sum.
\begin{align*}
\sum_{i=1}^k\sum_{j= i}^k
\frac{\log_2(\val_i(\sigs))-\log_2(\val_j(\sigs))}{j(j+1)} &=  \sum_{i=1}^k\log_2(\val_i(\sigs))\sum_{j=i}^k \frac{1}{j(j+1)} - \sum_{j=1}^k\sum_{i=j}^k \frac{\log_2(\val_i(\sigs))}{i(i+1)}\notag\\
&= \sum_{i=1}^k\log_2(\val_i(\sigs))\cdot \left(\frac{1}{i}-\frac{1}{k+1}\right) -\sum_{i=1}^k  \frac{\log_2(\val_i(\sigs))}{i(i+1)}\cdot i\\
&= \sum_{i=1}^k\log_2(\val_i(\sigs))\cdot \left(\frac{1}{i}-\frac{1}{k+1}-\frac{1}{i+1} \right)\\
&= \sum_{i=1}^k\log_2(\val_i(\sigs))\cdot\left(\frac{1}{i(i+1)}-\frac{1}{k+1}\right).
\end{align*}
Plugging back into \Cref{eq:bi_bound}, we get
\begin{align*}
    \sum_{i=1}^n B_i &= \frac{k}{k+1}(1-\log_2(\val_1(\sigs)))+\sum_{i=1}^k \log_2(\val_i(\sigs))\cdot\left(\frac{1}{k+1} + \frac{1}{i(i+1)}-\frac{1}{k+1} \right) \\
    &= \frac{k}{k+1}(1-\log_2(\val_1(\sigs)))+\sum_{i=1}^k \log_2(\val_i(\sigs))\cdot\frac{1}{i(i+1)} \\
    &\le \frac{k}{k+1}-\log_2(\val_1(\sigs))\cdot\frac{k}{k+1} + \log_2(\val_1(\sigs))\sum_{i=1}^k \frac{1}{i(i+1)} \\
    &= \frac{k}{k+1} \le 1.
\end{align*}
\end{proof}

Finally, we prove the upper bound on the probability of being a candidate which is used in \Cref{lem:dsb_feasibility}.
\begin{lemma}\label{lem:proba-candidate}
The probability of each bidder $i$ being a candidate is bounded by
$$
\E_{r,\pi}[c_i] \leq
\frac{1}{i(i+1)}+
\frac{\logdag(2\val_i(\sigs)/\low{1}{i})}{k+1}
+\sum_{j\in[k]\setminus\{i\}}\frac{\logdag(\val_i(\sigs)/\low{j}{i})}{j(j+1)}.
$$
\end{lemma}
\begin{proof}
Fix the random choices $r$ and $\pi$ of Mechanism~\ref{alg:SB}.
{We observe that the following conditions are equivalent for bidder $i$ to a candidate:}
\begin{align*}
c_i = 1 &\quad\Leftrightarrow\quad
\forall j\neq i, \quad f_r(\val_i(\sigs)) >_\pi f_r(\low{j}{i})
&(\text{by definitions of $c_i$})
\\
&\quad\Leftrightarrow\quad
\forall j\neq i, \quad \begin{cases}
f_r(\val_i(\sigs)) \geq f_r(\low{j}{i}) & \text{if }\pi(i) > \pi(j)\\
f_r(\val_i(\sigs)) > f_r(\low{j}{i}) & \text{if }\pi(i) < \pi(j)\\
\end{cases}&(\text{by definition of }>_\pi)
\\
&\quad\Leftrightarrow\quad
\forall j\neq i, \quad \begin{cases}
f_r(\val_i(\sigs)) > f_r(\low{j}{i}/2) & \text{if }\pi(i) > \pi(j)  \\
f_r(\val_i(\sigs)) > f_r(\low{j}{i}) & \text{if }\pi(i) < \pi(j) \\
\end{cases}&(\text{by definition of }f_r)
\end{align*}
To simplify this condition with two cases, we sort agents by decreasing low estimates, and we let $\sigma(\ell)$ denote the bidder $j$ with $\ell$-th highest $\low{j}{i}$. In particular, we have
$$
\low{\sigma(1)}{i} \geq \low{\sigma(2)}{i} \geq \dots \geq \low{\sigma(n-1)}{i}.$$
Next, we define
$$\tau_{i,\ell} = \max(\low{\sigma(\ell)}{i}, \low{\sigma(1)}{i}/2)
\quad\text{and}\quad
\tau_{i,n} = \low{\sigma(1)}{i}/2.$$
Finally, we let $t(\pi) = \min\{\ell\in[n-1]\;|\;\pi(i) < \pi(\sigma(\ell))\}$
if $\pi$ ranks some $j$ above $i$, and $t(\pi) = n$ otherwise. {Recall that if $i$ is a candidate then $f_r(\val_i(\sigs)) \ge f_r(\low{j}{i})$ for all $j\neq i$. Hence, observe that if $i$ is a candidate then  $f_r(\val_i(\sigs)) \ge f_r(\low{\sigma(1)}{i}) > f_r(\low{\sigma(1)}{i}/2)$, no matter whether $\pi(\sigma(1)) > \pi(i)$ or $\pi(\sigma(1)) < \pi(i)$. }
This gives the simplified condition
$$
c_i = 1 \quad\Leftrightarrow\quad f_r(\val_i(\sigs)) > f_r(\tau_{i,t(\pi)}).
$$
We compute the expected value of $c_i$
\begin{align*}
\E_{r,\pi}[c_i] &=
\sum_{\ell=1}^{n} \Pr_\pi[t(\pi) = \ell]\cdot \E_{r,\pi}[c_i\,|\,t(\pi) = \ell]
&(\text{law of total probability)}\\
&=
\sum_{\ell=1}^{n} \Pr_\pi[t(\pi) = \ell]\cdot \Pr_r[f_r(\val_i(\sigs)) > f_r(\tau_{i,\ell})] &(\text{condition above})\\
&=
\sum_{\ell=1}^{n} \Pr_\pi[t(\pi) = \ell]\cdot \logdag(\val_i(\sigs)/\tau_{i,\ell}). &(\text{using \Cref{lem:rounded-threshold}}).
\end{align*}

Now, remains to compute the probability that $t(\pi)=\ell$, induced by the uniformly random ordering $\pi$.
Observe that $t(\pi) > \ell$ if and only if $i$ is ranked before all bidders $\sigma(1), \dots, \sigma(\ell)$, which happens with probability $\frac{1}{\ell+1}$, for any $\ell\in [n-1]$. In particular, this implies that $t(\pi)=\ell$ with probability $\frac{1}{\ell}-\frac{1}{\ell+1} = \frac{1}{\ell(\ell+1)}$, hence
\begin{equation}
\E[c_i] = \frac{\logdag(\val_i(\sigs) /\tau_{i,n})}{n} + \sum_{\ell=1}^{n-1} \frac{\logdag(\val_i(\sigs)/\tau_{i,\ell})}{\ell(\ell+1)}.\label{eq:proba}
\end{equation}
To conclude the proof, we will permute terms in the sum using the rearrangement inequality\footnote{The rearrangement states that for every choice of real numbers $x_1\le \ldots\le x_n$, {$y_1\ge \ldots \ge y_n$} and permutation $\rho:[n]\rightarrow[n]$, we have $x_n y_1+ \ldots + x_1 y_n \le x_{\rho(1)}y_1 +\ldots + x_{\rho(n)}y_n$.}. We define $\sigma(n) = i$, such that $\sigma:[n]\rightarrow[n]$ is a permutation, and we denote $\sigma^{-1}$ its inverse. \begin{align*}
\E_{r,\pi}[c_i] &=
\frac{\logdag(\val_i(\sigs)/\tau_{i,n})}{n+1} +
\sum_{\ell\in[n]}
\frac{\logdag(\val_i(\sigs)/\tau_{i,\ell})}{\ell(\ell+1)}
&({\textstyle\frac{1}{n} = \frac{1}{n+1}+\frac{1}{n(n+1)}})\\
&\leq
\frac{\logdag(\val_i(\sigs)/\tau_{i,n})}{n+1}
+ \sum_{j\in[n]}\frac{\logdag(\val_i(\sigs)/\tau_{i,\sigma^{-1}(j)})}{j(j+1)}
&\text{(rearrangement ineq.)}\\
&\leq
\frac{\logdag(\val_i(\sigs)/\tau_{i,n})}{k+1}
+ \sum_{j\in[k]}\frac{\logdag(\val_i(\sigs)/\tau_{i,\sigma^{-1}(j)})}{j(j+1)}
&\text{($\tau_{i,n} \leq \tau_{i,\sigma^{-1}(j)}$ for all $j$)}\\
&\leq
\frac{\logdag(\val_i(\sigs)/\tau_{i,n})}{k+1}
+\frac{1}{i(i+1)}
+ \sum_{j\in[k]\setminus\{i\}}\frac{\logdag(\val_i(\sigs)/\low{j}{i})}{j(j+1)}
&\text{($\low{j}{i} \leq \tau_{i,\sigma^{-1}(j)}$)}\\
&\leq
\frac{\logdag(2\val_i(\sigs)/\low{1}{i})}{k+1}
+\frac{1}{i(i+1)}
+ \sum_{j\in[k]\setminus\{i\}}\frac{\logdag(\val_i(\sigs)/\low{j}{i})}{j(j+1)}.
&\text{($\low{1}{i}/2 \leq \tau_{i,n}$)}
\end{align*}
The last inequality corresponds to the statement of the Lemma.
\end{proof}

\subsection{Polynomial Time Implementation}

In this section, we show how to implement Mechanism~\ref{alg:SB} in polynomial time. The main technical challenge is to avoid enumerating all possible tie breaking permutation $\pi$ in Mechanism~\ref{alg:SB}. The polynomial time implementation is illustrated in Mechanism~\ref{alg:SBpoly}. {The mechanism makes $n^2$ queries to bidders. In general, it queries  bidders for their low estimates, that is $\{\low[\hat]{i}{j}\}_{i,j}$. Note that when valuation functions are monotone, it suffices to query valuations on the minimum signals, that is $\{\val[\hat]_i(0_j,\sigs[\hat]_{-j})\}_{i,j}$.} The mechanism queries bidders for their value on polynomially many signal profiles, which relates to the different values each bidder's value has to pass in order for the bidder to be a candidate, the thresholds $\tau_{i,\ell}$. It then computes the probability of each bidder to be a candidate using the $\logdag$ function, and the corresponding payment using \Cref{eq:pricing}.

\begin{algorithm}[p]
    \begin{enumerate}
\item Elicit reported signals $\sigs[\hat] = \{\sig[\hat]_i \in S_i\}_{i\in [n]}$, and query each bidder $i$ on:\begin{itemize}
    \item its value $\val[\hat]_i(\sigs[\hat])$ for signal profile $\sigs[\hat]$; and
    \item for every bidder $j\neq i$, query $i${'s} lowest possible value $\low[\hat]{{i}}{{j}}$ for signals $\sigs[\hat]_{-i}$.
\end{itemize}   
\item Let $\eta \ge 1$ be identical to $\eta$ in Mechanism~\ref{alg:SB}. 
\item For each bidder $i$, we define the following:
\begin{itemize}
\item{Let $\tau_{i,n} = \max_{j\neq i} \low[\hat]{j}{i}/2$, and let $\tau_{i,\ell}$ be the $\ell$-th value in $\{\max(\low[\hat]{j}{i}, \tau_{i,n})\}_{j\neq i}$.}
\item{Let $\displaystyle x_i = \frac{1}{\eta}\left(\frac{\logdag(\val[\hat]_i(\sigs[\hat])/\tau_{i,n})}{n} + \sum_{\ell=1}^{n-1} \frac{\logdag(\val[\hat]_i(\sigs[\hat])/\tau_{i,\ell})}{\ell(\ell+1)}\right).$}
\item{Let $\displaystyle p_i = \frac{1}{\eta}\left(\frac{\max(0, \min(\tau_{i,n}, \val[\hat]_i(\sigs[\hat])-\tau_{i,n}))}{n\ln 2} + \sum_{\ell=1}^{n-1} \frac{\max(0, \min(\tau_{i,\ell}, \val[\hat]_i(\sigs[\hat])-\tau_{i,\ell}))}{\ell(\ell+1)\ln 2}\right).$}
\end{itemize}
\item Allocate the item to bidder $i$ with probability $x_i$ for all $i\in [n]$.
\item Charge price $p_i$ to bidder $i$.
\end{enumerate}
\renewcommand{\thealgorithm}{PRCF}
\caption{\textbf{Polynomial-time Randomized  Candidate Filtering (PRCF) Mechanism}.}
\label{alg:SBpoly}
\end{algorithm}

\Cref{fig:polytime} illustrates some of the components used in the proof of \Cref{prop:polytime}. 

\begin{lemma} \label{prop:polytime}
    Mechanism~\ref{alg:SB} can be implemented in polynomial-time.
\end{lemma}

\begin{figure}[p]
    \centering
    \includegraphics[width=\textwidth]{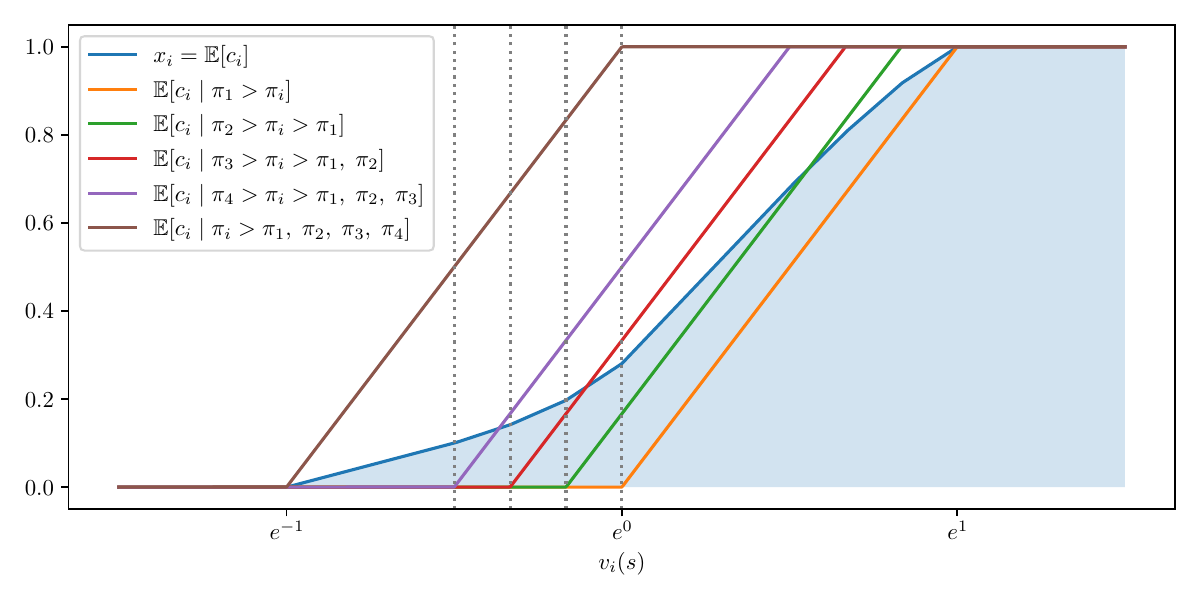}
    \caption{The blue piece-wise linear curve denotes the log-scale plot of the allocation probability $x_i$ as a function of $v_i(\sigs)$. In particular, $x_i$ is the weighted average of $\E_{r,\pi}[c_i|t(\pi) = \ell]$ for all $\ell$, weighted by the probability that $t(\pi) =\ell$. The log-scale plot of $\E_{r,\pi}[c_i|t(\pi) = \ell]$ is denoted by the different colored piece-wise linear curves. The lower estimates of bidders are represented with dashed lines at $\low{1}{i} = 1$, $\low{2}{i} = e^{-1/6}$, $\low{3}{i} = e^{-1/3}$ and $\low{4}{i} = e^{-1/2}$.}
    \label{fig:polytime}
\end{figure}

\begin{proof}
We show that Mechanism~\ref{alg:SBpoly} is a polynomial time implementation of Mechanism~\ref{alg:SB}. {First, Mechanism~\ref{alg:SB} is truthful, thus we assume bidders bid their true valuations and signals, and write $\sigs$ and $\vals$ instead of $\sigs[\hat]$ and $\vals[\hat]$.}

The mechanism queries the bidder's valuations on $O(n^2)$ many signal profiles. In particular, after eliciting signals $\sigs$, each bidder $i$ is then asked to report $\val_i(\sigs)$ and $\low{{i}}{{j}}$ for all $j\neq i$. For extending the mechanism to the case where $d$ is unknown (as discussed in the next section), we also ask the bidders to report the minimum $d_i$ such that their valuation function $v_i(\cdot)$ is $d_i$-\selfbounding. 
The mechanism runs in polynomial time as it gives a tractable formula to compute each bidder's allocation probability and payment as a function of $n$ different thresholds $\tau_{i,\ell}$. {First, notice that the probabilities in mechanisms \ref{alg:SB} and \ref{alg:SBpoly} are equal, as demonstrated in \Cref{eq:proba}.}

It remains to show that the payment formula implements \Cref{eq:pricing} for the given allocation rule.  
We first make the following observations towards computing $\logdag(v/c)\cdot v -\int_0^ v \logdag(t/c) \mathrm{d}t$ for all $v \ge 0$ and constant $c>0$. 
\begin{itemize}
    \item If $v \le c$, then $\int_0^ v \logdag(t/c) \mathrm{d}t = 0$.
    \item If $c \le v \le 2c$, then $\int_0^ v \logdag(t/c) \mathrm{d}t = 0 + \int_c^v \log_2(t/c) \mathrm{d}t = v{\log_2(v/c)}- \frac{v}{\ln 2} + \frac{c}{\ln 2}$.
    \item If $v \ge 2c$, then \begin{align*}
    \int_0^v \logdag(t/c) \mathrm{d}t &=\int_c^{2c} \log_2 (t/c) \mathrm{d}t + \int_{2c}^v \mathrm{d}t\\
    &= \left(2c\log_2(2c/c) - \frac{2c - c}{\ln 2} \right) + (v - 2c)\\
    & = v - \frac{c}{\ln 2}
\end{align*}
\end{itemize}

Plugging in the above observations we immediately get
\begin{align}
\logdag(v/c)\cdot v - \int_0^v \logdag(t/c) \mathrm{d}t &= 
\begin{cases}
    0, &\text{if } v \le c\\
    {(v - c)}/{\ln 2}, & \text{if } c < v \le 2c\\
    c/\ln 2, &\text{if } v > 2c
\end{cases}\notag\\
& = \frac{\max(0,\min(c, v-c))}{\ln 2}\label{eq:price of logdag}
\end{align}

We are now ready to compute the payments for our allocation rule $x_i$ according to \Cref{eq:pricing}. For any fixed $\vals_{-i}$, $\sigs_{-i}$, and for all $v \ge 0$ we have
\begin{align*}
    &x_i(v)\cdot v - \int_0^v x_i(t) \mathrm{d}t \\
    & = \frac{1}{\eta}\left(\frac{\logdag(v/\tau_{i,n})}{n} + \sum_{\ell=1}^{n-1} \frac{\logdag(v/\tau_{i,\ell})}{\ell(\ell+1)}\right)\cdot v - \int_0^v \frac{1}{\eta}\left(\frac{\logdag(t/\tau_{i,n})}{n} + \sum_{\ell=1}^{n-1} \frac{\logdag(t/\tau_{i,\ell})}{\ell(\ell+1)}\right)\mathrm{d}t \\
    &= \frac{1}{\eta} \left(\frac{\logdag(v/\tau_{i,n})\cdot v}{n} -  \int _0^ v \frac{\logdag(t/\tau_{i,n})}{n}\mathrm{d}t + \sum_{\ell =1}^{n-1}\left(\frac{\logdag(v/\tau_{i,\ell})\cdot v}{\ell(\ell+1)} - \int_0^t\frac{\logdag(t/\tau_{i,\ell})}{\ell(\ell+1)}\mathrm{d}t\right)\right) \\
    &= \frac{1}{\eta}\left(\frac{\max(0,\min(\tau_{i,n},v-\tau_{i,n}))}{n\ln 2}+ \sum_{\ell =1}^{n-1}\frac{\max(0,\min(\tau_{i,\ell},v-\tau_{i,\ell})}{\ell(\ell+1)\ln 2}\right), \qquad \text{(By Eq~\eqref{eq:price of logdag})}
\end{align*}
which is exactly the payment $p_i$ for $\val_i(\sigs) = v$ defined in Mechanism~\ref{alg:SBpoly}.

\end{proof}

\subsection[Unknown d]{Unknown $d$}\label{sec:unknown-d}

In this section we show how to extend our results for the case where the value of $d$ is unknown. Recall that Mechanism \ref{alg:SB} uses a bound on $d$ to set the normalization factor $\eta$. 
However, in order to keep the mechanism truthful, the allocation of $i$ cannot rely on $i$'s valuation function $\val_i(.)$, and hence on $d_i$.
{To address this challenge, we define personalized normalization factors for each bidder $i$ that doesn't use $d_i$.}
In particular, for each bidder $i$, we use the smallest value $d$ such that all bidders $j\neq i$ are $d$-\selfbounding. 
This way, all bidders except at most a single bidder (the one with the largest $d$) have the correct normalization factor $\eta$. 
Thus, by scaling the allocation probabilities by another factor of $2$, feasibility is guaranteed and the $O(d)$-approximation is preserved. We note that {this modification requires asking each bidder to report the smallest value $d_i$ such that their valuation function is $d_i$-\selfbounding.}

This is formalized in the following lemma.

\begin{lemma}\label{prop:unknown_d_sb_feasibility}
    For every instance with $d$-\selfbounding valuations, where $d$ is unknown, one can compute personalized normalization factors $\eta_i$ such that by setting $x_i = \E_{r,\pi}[c_i]/\eta_i$ in the~\ref{alg:SB} mechanism we get a feasible, truthful, $O(d)$-approximation to the optimal welfare.
\end{lemma}

\begin{proof}
Given reported signals $\sigs[\hat]$ and valuations $\vals[\hat]$, we set $\eta_i = 4(\hat{d}_{-i} + 1)$ with $\hat{d}_{-i} = \max_{j\neq i} d_j$ where each $\val[\hat]_j$ is $d_j$-\selfbounding (which means $d_j \le d$). Observe that $\hat{d}_{-i}$ doesn't depend on $\hat{s}_i$ or $\val[\hat]_i$. Moreover, by the same arguments as \Cref{lem:dsb_truthful},  $c_i$ is monotone in $\val[\hat]_i(\sigs[\hat])$ and doesn't depend on $i$'s information. This implies that the allocation $x_i$ is also monotone in $\val[\hat]_i(\sigs[\hat])$ and doesn't depend on $\hat{s}_i$ or $\val[\hat]_i$. Hence, by \Cref{lem:pricing} the allocation is EPIC-IR implementable.

Next, we show that the resulting allocation is feasible. Recall that in~\Cref{lem:dsb_feasibility} we showed that when $\eta = 2(d+1)$ we have $\sum_i \E_{r,\pi}[c_i]/\eta \le 1$. Suppose $\eta_i = \eta_j$ for all bidders $i,j$, we observe that $\eta_i = 2\eta$. This is because we have $\hat{d}_{-i} = \hat{d}_{-j} = \max_i d_i = d$. Hence by~\Cref{lem:dsb_feasibility} we have $\sum_i x_i = \sum_{i} \E[c_i]/2\eta \le 1/2$, which is a feasible allocation. Suppose $\eta_i \neq \eta_j$ for some bidders $i \neq j$, then we have a unique bidder $i = \argmax_j d_j$ with $d_i = d$. Hence we have $\eta_j = 2\eta$ for all $j \neq i$, and by~\Cref{lem:dsb_feasibility} we have $\sum_{j\neq i} x_i \le 1/2$. Moreover, since $\eta_i \ge 2$, we have $x_i \le 1/2$. This implies $\sum_j x_j \le 1$, and thus proving feasibility.

Finally, we show that the we obtain a $O(d)$-approximation. Observe that $\eta_i \le 2\eta$ for all bidders $i$, so by the same arguments \Cref{lem:dsb_approx} we have that the mechanism is a $2\eta \cdot 2\ln 2$ approximation to the optimal welfare.

\end{proof}

\subsection{Putting it All Together}
We now have all the ingredients to prove our main theorem. The full proof follows.

\begin{proof}[Proof of \Cref{thm:sos-constant}]
\Cref{lem:dsb_truthful} shows the mechanism is truthful, \Cref{lem:dsb_approx} and \Cref{lem:dsb_feasibility} show the mechanism is feasible and gets $4(d+1)\ln 2=O(d)$-approximation for $d$-\selfbounding valuations. By \Cref{cor:d_sb_impossibility}, every EPIC-IR mechanism cannot have a better than $\Omega(d)$-approximation for $d$-\selfbounding valuations, even if the valuations are public. By \Cref{lem:sos is sb}, this gives a $8\ln 2\approx 5.55$-approximation for monotone \sos functions, and a $12\ln 2\approx 8.32$-approximation for non-monotone \sos functions. By  \Cref{prop:unknown_d_sb_feasibility}, all the results generalize to the case where the bound on $d$ is unknown by losing another factor of $2$ in the approximation. Finally, by \Cref{prop:polytime}, the mechanism can be implemented in polynomial time.     
\end{proof}


\section{Multi-Unit Auctions with Unit-Demand Bidders}\label{sec:multi-unit}

In this section we extend results from \Cref{sec:d-sb-mech} to multi-unit auctions with $n$ unit-demand bidders and $m$ identical items. We assume that $1 \leq m < n$; otherwise, trivially we could give each bidder one of the items.  We consider the following small adjustment of \ref{alg:SB} and show that this gives a truthful and feasible mechanism that obtains an $O(d)$-approximation when allocating $m$ identical items to $n$ unit-demand bidders.

\paragraph{Adjusted ~\ref{alg:SB} Mechanism:}
    \begin{enumerate}
        \item In step (3) of the mechanism we set $c_i = 1$ if $f_r(\val_i(\sigs)) >_\pi f_r(\low{j}{i}) $ for at least $n-m$ bidders $j\neq i$, and $0$ otherwise.
        \item In step (4) of the mechanism we allocate items using \Cref{thm:birkhoff} such that the allocation is ex-post feasible.
    \end{enumerate}

\begin{theorem}\label{thm:multi_sb}
    Given an instance with $n$ unit-demand bidders and $m$ identical items, there exist an EPIC-IR mechanism which obtains $O(d)$-approximation for $d$-\selfbounding valuations.
\end{theorem}

\begin{proof}
   We first observe that the resulting mechanism is truthful following the same arguments as \Cref{lem:dsb_truthful}.
    Next, in \Cref{lem:dsb-fractional-feasibility} we claim that for $\eta = 4(d+1)$ the resulting allocation is fractionally feasible, i.e., $\sum_i{x_i} \le m$. We further use a randomized rounding procedure following \Cref{thm:birkhoff} to obtain a randomized allocation such that the each bidder $i$ is allocated an item with probability $x_i$, while making sure the allocation  is ex-post feasible, that is, at most $m$ bidders are allocated. 

    Finally, we show that the mechanism obtains an $O(d)$-approximation to the optimal welfare. Fix valuations and signals $\vals,\sigs$ and consider any random choice of $r$ and $\pi$. Wlog we rename bidders such that $\val_1(\sigs) \ge \val_2(\sigs)\ge \cdots \ge \val_n(\sigs)$.
    Let $I^\star$ denote the top $m$ bidders according to $f_r(\val_i(\sigs))$ breaking ties according to priority in $\pi$. We denote $I^\star = \{i_1,i_2,\cdots,i_m\}$ where $$f_r(\val_{i_1}(\sigs)) >_\pi f_r(\val_{i_2}(\sigs)) >_\pi\ldots  >_\pi f_r(\val_{i_m}(\sigs)).$$
    For each $i_\ell\in I^\star$ it must be the case that $c_{i_\ell} = 1$. This is because
    $$f_r(\val_{i_\ell}(\sigs)) >_\pi f_r(\val_{j}(\sigs)) \ge f_r(\low{j}{i})$$ for every $j\notin I^\star$ (thus, for at least $n-m$ many bidders, which implies that $i_\ell$ is a candidate).
    Moreover, we have that
 \[
 \val_{i_\ell}(\sigs) \ge f_r(\val_{i_\ell}(\sigs)) \ge f_r(\val_{\ell}(\sigs)).
 \]
 Therefore, for every $\vals,\sigs,r,\pi$ we have that $$\sum_{i=1}^n c_i \val_i(\sigs) \ge \sum_{i_\ell\in I^\star} \val_{i_\ell}(\sigs) \ge \sum_{\ell=1}^m f_r(\val_{\ell}(\sigs)).$$
 Hence, we have
 \begin{align*}
     \sum_{i=1}^n x_i \cdot \val_i(\sigs) & = \sum_i\frac{\E[c_i]}{\eta}\cdot \val_i(\sigs)
     \ge {\E\left[\sum_{i=1}^n \frac{c_i \val_i(\sigs)}{\eta}\right]}
     \ge \E\left[\sum_{\ell =1}^m \frac{f_r(\val_{\ell}(\sigs))}{\eta}\right]
     \ge \sum_{\ell =1}^m\frac{\val_{\ell}(\sigs)}{\eta\cdot 2\ln 2},
 \end{align*}
 where for the last inequality we recall, from the proof of \Cref{lem:dsb_approx}, that $\E_r[f_r(v)] \ge v/(\eta\cdot 2\ln 2)$. For $\eta = 4(d+1)$ this provides the desired $O(d)$-approximation.
 \end{proof}

We also observe that the same extension can be made to the case the mechanism is oblivious of the value $d$.
 \begin{observation}
     The adjusted \ref{alg:SB} mechanism can be made oblivious to $d$ by losing an additional factor of $2$, by using personalized normalization parameters $\eta_i = 4\max_{j\neq i}(d_j+1)$.
 \end{observation}

\subsection{Fractional Feasibility}
In this section we show that the adjusted~\ref{alg:SB} provide a fractionally feasible allocation for any multi-unit auction instance. In particular, we show that the allocation probabilities, $x_i$, always sum up to at most $m$. 

\begin{lemma}\label{lem:dsb-fractional-feasibility}
Let $x_i = \E[c_i]/\eta$ be the allocation probability from the \emph{adjusted} \ref{alg:SB} with $\eta = 4(d+1)$. Then the allocation is fractionally feasible, that is,
\[
\sum_{i=1}^n x_i \le m.
\]
\end{lemma}

\begin{proof}
We will show that the expected number of candidates $\sum_{i}\E[c_i]$ is at most ${4}m(d + 1)$. Thus for $\eta = {4}(d +1)$ we have $\sum_i x_i \le m$ as desired.
Wlog we rename the bidders such that $v_1(\sigs) \ge v_2(\sigs) \ge \ldots \ge v_n(\sigs)$. Let $k$ be the number of bidders whose values are larger that $v_m/2$. Similar to the single item settings, we distinguish the bidders as large valued (numbered $1$ through $m-1$), intermediate valued (numbered $m$ through $k$), and small valued (numbered $k+1$ through $n$) for the analysis. 

We first consider the highest $m-1$ bidders. Since the probability that each one of them is a candidate is at most $1$, we get $\sum_{i=1}^{m-1}\E[c_i] \le m-1$. 

Second, we consider the small valued bidders $i\in\{k+1, \dots, n\}$. 
We observe that for each one of them to be a candidate it is necessary that there are at most $m-1$ bidders $j\neq$ i such that $f_r(\low{j}{i}) > f_r(v_i(\sigs))$. Hence, there exists some $j\in [m]$ such that $f_r(v_i(\sigs)) \ge f_r(\low{j}{i}) > f_r(\low{j}{i}/2)$. With this we bound the probability that a small valued bidder $i$ is a candidate as follows,
\begin{align}
    \E_{r,\pi}[c_i] &\le \Pr[\exists j \in [m] \text{ such that } f_r(v_i(\sigs)) > f_r(\low{j}{i}/2)]\notag \\
    &\le \sum_{j=1}^m \logdag\left(\frac{2v_i(\sigs)}{\low{j}{i}}\right) \le 
    \sum_{j=1}^m \logdag\left(\frac{v_j(\sigs)}{\low{j}{i}}\right),
\end{align}
where the last inequality follows by definition of small valued. Using \Cref{lem:dsb}, it follows that $\sum_{i=k+1}^n \E[c_i] \le 2dm$.

Finally, it remains to bound the probability that intermediates bidders $i\in \{m, \dots, k\}$ are candidates.
In \cref{lem:gen_candidate_prob} (below) we show that,
    \[
    \E_{r,\pi}[c_i] \le \frac{m}{i(i+1) } + \frac{m}{k+1} + \sum_{j=1}^k\frac{m}{j(j+1)}\cdot\logdag\left(\frac{v_i(\sigs)}{\low{j}{i}}\right).
    \]
    By \Cref{obs:logdag-prod} we get the following bound,
    \begin{equation}
\E_{r,\pi}[c_i] \leq 
\frac{m}{i(i+1)} + \frac{m}{k+1} 
+\underbrace{\sum_{j=1}^k\frac{m}{j(j+1)}\cdot\logdag\left(\frac{v_j(\sigs)}{\low{j}{i}}\right)}_{A_i}
+\underbrace{\sum_{j=1}^k\frac{m}{j(j+1)}\cdot\logdag\left(\frac{v_i(\sigs)}{v_j(\sigs)}\right)}_{B_i}
\label{eq:multi-decomp}
\end{equation}
{We observe that
\begin{align*}
\sum_{i=m}^{k}\left(\frac{m}{i(i+1)} +\frac{m}{k+1}\right)
&= 1 - \frac{m}{k+1} + m\cdot\frac{k+1-m}{k+1}\le 1 + m.
\end{align*}
We next show that $\sum_{i=m}^k A_i \le 2dm$ (in \Cref{lem:multi_ai_bound}), and $\sum_{i=m}^k B_i \le m$ (in \Cref{lem:multi_bi_bound}), implying that $\sum_{i=m}^k  \E[c_i] \le 2dm+2m+1$.}

Overall, we get $\sum_{i=1}^n \E[c_i] \le 4dm+3m\leq 4(d+1)m$.
\end{proof}

\begin{lemma}\label{lem:multi_ai_bound}
    Given an instance with $d$-\selfbounding valuations, we have that
    \[\sum_{i=m}^k A_i \le 2dm\]
    where $A_i$'s are defined in the proof of \Cref{lem:dsb-fractional-feasibility}.
\end{lemma}
\begin{proof}
 Recall definition of $A_i$ for large valued bidders $i\in [k]$ from \Cref{eq:multi-decomp},
    \[
    A_i = \sum_{j=1}^k\frac{m}{j(j+1)}\cdot\logdag\left(\frac{v_j(\sigs)}{\low{j}{i}}\right),
    \]
    Hence, summing over $i$ and swapping the summation of $i$ and $j$ we get,
    \begin{align*}
        \sum_{i=1}^n A_i & = \sum_{j=1}^k\sum_{i=1}^n \frac{m}{j(j+1)}\logdag\left(\frac{v_j(\sigs)}{\low{j}{i}}\right)\le \sum_{j=1}^k\frac{m}{j(j+1)} 2d
        = 2dm\cdot\left(1 - \frac{1}{k+1}\right) \le 2dm,
    \end{align*}
    where the inequality follows from \Cref{lem:dsb}.
\end{proof}

\begin{lemma}\label{lem:multi_bi_bound}
    Given an instance, we have that
    $$
    \sum_{i=m}^k B_i \le m
    $$
    where $B_i$'s are defined in the proof of \Cref{lem:dsb-fractional-feasibility}.
\end{lemma}
\begin{proof}
We first re-write the sum by first observing that for all $j<i$ we have $\logdag(v_i(\sigs)/v_j(\sigs))$ is $0$ and group together all terms corresponding to any bidder $i$.
    \begin{align*}
    \sum_{i=m}^k\sum_{j= 1}^k
   \frac{m}{j(j+1)}\cdot\logdag\left(\frac{v_i(\sigs)}{v_j(\sigs)}\right)
    &\le\sum_{i=m}^k\sum_{j= i}^k
   \frac{m}{j(j+1)}\cdot(\log_2(\val_i(\sigs))-\log_2(\val_j(\sigs))) \\
   &=  \sum_{i=m}^k\log_2(\val_i(\sigs))\sum_{j=i}^k \frac{m}{j(j+1)} - \sum_{j=m}^k\sum_{i=j}^k \frac{m}{i(i+1)}\cdot\log_2(\val_i(\sigs))\notag\\
    &= \sum_{i=m}^k\log_2(\val_i(\sigs))\cdot \left(\frac{m}{i}-\frac{m}{k+1}\right) -\sum_{i=m}^k   \frac{(i-m+1) m}{i(i+1)}\cdot\log_2(\val_i(\sigs))\\
    &= \sum_{i=m}^k\log_2(\val_i(\sigs))\cdot \left(\frac{m}{i}-\frac{m}{k+1}-\frac{m}{i} + \frac{m^2}{i(i+1)}\right)\\
    &= \sum_{i=m}^k\log_2(\val_i(\sigs))\cdot\left(\frac{m^2}{i(i+1)}-\frac{m}{k+1}\right)
    \end{align*}
We then upper bound (and lower bound) all $v_i(\sigs)$ by $v_m(\sigs)$ (and $v_k(\sigs)$ respectively) for all large valued bidders $i\ge m$, and note that $v_m(\sigs)\le 2v_k(\sigs)$ in order to obtain the desired inequality.  
    \begin{align*}
    \sum_{i=m}^k\sum_{j= 1}^k
   \frac{m}{j(j+1)}\cdot\logdag\left(\frac{v_i(\sigs)}{v_j(\sigs)}\right)
   &\le \sum_{i=m}^k\log_2(\val_i(\sigs))\cdot\left(\frac{m^2}{i(i+1)}-\frac{m}{k+1}\right)\\
    &\le \sum_{i=m}^k\log_2(\val_m(\sigs))\cdot \frac{m^2}{i(i+1)} - \sum_{i=m}^k\log_2(\val_k(\sigs))\cdot\frac{m}{k+1}\\
    &= \log_2(\val_m(\sigs))\cdot\left(\frac{m^2}{m} - \frac{m^2}{k+1}\right) - \log_2(\val_k(\sigs))\cdot\frac{m}{k+1}\cdot(k+1-m)\\
    &= \log_2(\val_m(\sigs))\cdot m\left(1 - \frac{m}{k+1}\right) - \log_2(\val_k(\sigs))\cdot m\left(1 - \frac{m}{k+1}\right) \\
    &= m\cdot (\log_2(\val_m(\sigs)) - \log_2(\val_k(\sigs)))\cdot\left(1 - \frac{m}{k+1}\right) \\
    &\le m
   \end{align*}
\end{proof}
    
\begin{lemma}\label{lem:gen_candidate_prob}
      For any $i$, the probability that $i$ is a candidate  (i.e., $c_i = 1$) is \[
    \E_{r,\pi}[c_i] \le \frac{m}{i(i+1)} + \frac{m}{k+1}+ \sum_{j\in [k] \setminus i} \frac{m}{j(j+1)}\cdot \left(\frac{v_i(\sigs)}{\low{j}{i}}\right).
    \]
\end{lemma}
\begin{proof}
    For any choice of $r$ and $\pi$, we observe that the following conditions are all necessary for bidder $i$ to be a candidate in the adjusted RCF mechanism:
    \begin{enumerate}
        \item there exists at most $m-1$ bidders $j\neq i$ such that $f_r(\low{j}{i}) > f_r(v_i(\sigs))$, and
        \item if there are at least $m$ bidders $j$ with $\pi (j) > \pi(i)$, then there are at most $m-1$ of these bidders $j$ such that $f_r(\low{j}{i}) \ge f_r(v_i(\sigs))$.
    \end{enumerate}
    To simplify these conditions we introduce the following notations. We first order all the other bidder $j\neq i$ in decreasing order of $\low{j}{i}$ as follows \[\low{\sigma(1)}{i} \ge \low{\sigma(2)}{i} \ge \ldots \ge \low{\sigma(n-1)}{i}.\]
    We define the following ``critical thresholds'', 
    \[\tau_{i,\ell} = \max (\low{\sigma(\ell)}{i}), \low{\sigma(m)}{i})/2)~\forall \ell \le n-1 \quad \text{and} \quad \tau_{i,n} = \low{\sigma(m)}{i})/2.\]  
    Finally, for any permutation $\pi$, we define $t(\pi) = n$ if $\pi (i) > n-m$, otherwise we define $t(\pi) = \ell$ such that $\pi(\sigma(\ell)) > \pi (i)$ and there exists exactly $m-1$ many $\ell' < \ell$ such that $\pi(\sigma(\ell')) > \pi(i)$. 
    
    Therefore, the necessary conditions for $i$ to be a candidate can be simplified as,
    \[
    c_i =1 \implies f_r(v_i(\sigs)) > f_r(\tau_{i,t(\pi)}).
    \]
    Hence we have,
    \begin{align}
            \E_{r,\pi}[c_i] & =  \sum_{\ell=m}^{n} \Pr_\pi[t(\pi) = \ell]\cdot \E_{r,\pi}[c_i|t(\pi) =\ell]  \notag\\
            & \le \sum_{\ell=m}^{n} \Pr_\pi[t(\pi) = \ell]\cdot \Pr_r[f_r(v_i(\sigs)) > f_r(\tau_{i,\ell})] \label{eq:candidate-prob-decomposition}
    \end{align}
    Observe that $t(\pi) = n$ is the event that $i$ has top $m$ rank according to $\pi$, which happens with probability $m/n$. Moreover, for each $\ell \in \{ m, m+1,\ldots, n-1\}$, $t(\pi) = \ell $ is the event that $i$ is exactly ranked $m+1$ amongst the $\ell + 1$ bidders $\{i, \sigma(1),\sigma(2),\ldots,\sigma(\ell)\}$ and $\sigma(\ell)$ is in the top $m$ rank amongst the other $\ell$ bidders. Thus we have,
    \begin{align}
    \Pr[t(\pi) = \ell] = \frac{1}{\ell +1}\cdot \frac{m}{\ell} \quad \text{and} \quad  \Pr[t(\pi) = n] = \frac{m}{n} \notag
   \end{align}
Hence, plugging this into \Cref{eq:candidate-prob-decomposition} and using \Cref{lem:rounded-threshold} we get,
   \begin{align*}
            \E_{r,\pi}[c_i] 
            & \le \sum_{\ell=m}^{n} \Pr_\pi[t(\pi) = \ell]\cdot \logdag(v_i(\sigs)/\tau_{i,\ell})\\
            & \le \frac{m}{n}\cdot \logdag (v_i(\sigs)/\tau_{i,n}) + \sum_{\ell=m}^{n-1} \frac{m}{\ell(\ell+1)}\logdag(v_i(\sigs)/\tau_{i,\ell}) \\
            &= \frac{m}{n+1}\cdot \logdag (v_i(\sigs)/\tau_{i,n}) + \sum_{\ell=m}^{n} \frac{m}{\ell(\ell+1)}\logdag(v_i(\sigs)/\tau_{i,\ell}) &\left(\frac{m}{n} = \frac{m}{n+1} + \frac{m}{n(n+1)}\right)\\
            &\le \frac{m}{n+1}\cdot \logdag (v_i(\sigs)/\tau_{i,n}) + \sum_{\ell=1}^{n} \frac{m}{\ell(\ell+1)}\logdag(v_i(\sigs)/\tau_{i,\ell})
            \end{align*}
    Because $\sigma$ orders the bidders in decreasing order of lower estimates, we have that $\logdag(v_i(\sigs)/\tau_{i,\ell})$ are increasing in $\ell$. Moreover, since $1/(\ell(\ell+1))$ are decreasing in $\ell$, by the rearrangement inequality we have
            \begin{align*}
            \E_{r,\pi}[c_i] 
            & \le \frac{m}{n+1}\cdot \logdag  \left(\frac{v_i(\sigs)}{\tau_{i,n}}\right) + \sum_{j\neq i} \frac{m}{j(j+1)}\logdag\left(\frac{v_i(\sigs)}{\max(\low{j}{i},\tau_{i,n})}\right) + \frac{m}{(i(i+1))}{\logdag \left(\frac{v_i(\sigs)}{\tau_{i,n}}\right)},
            \end{align*}
where we reshuffle the $\logdag(v_i(\sigs)/\tau_{i,\ell}))$ terms and recall by definition  $\tau_{i,\ell} = \max(\low{\sigma(\ell)}{i},\tau_{i,n})$.

Next, we bound the $\logdag$ terms by using  $\tau_{i,n}$ for $j > k$ and $\low{j}{i}$ for $j\le k$  to obtain,
            \begin{eqnarray*}
            \E_{r,\pi}[c_i] &\le \left(\frac{m}{n+1}+\frac{m}{i(i+1)}\right)\cdot \logdag \left(\frac{v_i(\sigs)}{\tau_{i,n}}\right) &+ \sum_{j\in [k]\setminus \{i\}} \frac{m}{j(j+1)}\logdag\left(\frac{v_i(\sigs)}{\low{j}{i}}\right) \\
            & &+ \sum_{\substack{k<j\le n \\j\neq i}} \frac{m}{j(j+1)}\cdot \logdag  \left(\frac{v_i(\sigs)}{\tau_{i,n}}\right) \\
            &\le \left(\frac{m}{k+1}+\frac{m}{i(i+1)}\right)\cdot \logdag \left(\frac{v_i(\sigs)}{\tau_{i,n}}\right) &+ \sum_{j\in [k]\setminus \{i\}}\frac{m}{j(j+1)}\logdag\left(\frac{v_i(\sigs)}{\low{j}{i}}\right) 
            \end{eqnarray*}
\end{proof}

\subsection{Ex-post Feasibility} 
Finally, using a folklore randomized rounding procedure, that follows from Birkhoff decomposition, we obtain an ex-post feasible allocation where at most $m$ items are allocated while preserving the marginal probability of allocation for each bidder.

{Birkhoff decomposition states that doubly stochastic matrices (square matrices, with each row/columns summing to 1) are convex combinations of permutations matrices (with exactly one 1 per row/column). \Cref{thm:birkhoff} is a folklore generalization of Birkhoff decomposition, which we use to turn the probability vector $(x_1, \dots, x_n)\in [0,1]^n$ such that $\sum_{i=1}^n x_i \leq m$ into a distribution over feasible allocations such that each bidder $i$ either receives no items or a single item, and the marginal probability of receiving an item is exactly~$x_i$.}

\begin{proposition}
    \label{thm:birkhoff}
    Let $\mathcal M$ be the set of $n\times m$ matrices with non-negative values, such that each row and each column sums to at most 1. Any matrix in $\mathcal M$ can be decomposed (in polynomial time) {into} a convex combination of matrices from $\mathcal M$ with $\{0,1\}$ coefficients.
\end{proposition}

To prove~\Cref{thm:birkhoff}, we use the following folklore generalization of Kőnig's line coloring theorem~\cite{lovasz1986matching}.

\begin{proposition}
    \label{lem:konig}
    Given a weighted bipartite graph (positive edge weights) with at least one edge, there is a matching which covers all maximum-degree vertices (sum of weights of incident edges).
\end{proposition}

We now prove~\Cref{thm:birkhoff}.

\begin{proof}[Proof of \Cref{thm:birkhoff}]
Consider a matrix $M_0\in \mathcal M$. We see $M_0$ as a weighted bipartite graph, where nodes are rows and columns, and edges corresponds to cells with positive values. Let $\Delta(M_0)$ be the maximum degree of a vertex. If $\Delta(M_0) = 0$ the proof is finished. Otherwise, by \Cref{lem:konig}, there exists a matching $\mu_0 \in \mathcal M \cap \{0,1\}^{n\times m}$ which covers all maximum degree vertices, and which can be computed in polynomial time (with a maximum weight matching algorithm). {Let $v_0 > 0$ be the difference between the highest and second highest degrees, and let $w_0 > 0$ be the minimum weight of an edge in $\mu_0$. Define $z_0 = \min(v_0, w_0)$, add $z_0\cdot \mu_0$ to the decomposition, and define $M_1 = M_0 - z_0 \cdot\mu_0$. Notice that $M_1$ contains less edges (strictly, if $w_0\leq v_0$) and more maximal-degree vertices (strictly, if $v_0 \leq v_0$) than $M_0$. Additionally, $\Delta(M_1) = \Delta(M_0)-z_0$.} Apply inductively the same argument to define $M_1$, $M_2$, until reaching $M_k = 0$. We obtained a decomposition of $M_0$ as a positive linear combination of matchings, with a sum of coefficients equal to $\Delta(M_0) \leq 1$. We conclude by adding the empty matching with a coefficient of $1-\Delta(M_0)$.
\end{proof}

\bibliographystyle{plainnat}
\bibliography{focs-idv}

\end{document}